\documentclass[11pt]{article}
\usepackage{amssymb}
\usepackage{amsfonts,amsmath, longtable}
\usepackage{amstext, amsfonts,amsbsy}

        \def\theequation{\thesection.\arabic{equation}}

\newcommand{\Si}{\Sigma}
\newcommand{\tr}{{\rm tr}}

\newcommand{\Ad}{{\rm Ad}}
\newcommand{\ti}[1]{\tilde{#1}}
\newcommand{\om}{\omega}
\newcommand{\Om}{\Omega}
\newcommand{\de}{\delta}
\newcommand{\al}{\alpha}

\newcommand{\vth}{\vartheta}
\newcommand{\be}{\beta}

\newcommand{\ep}{\epsilon}
\newcommand{\vf}{\varphi}

\newcommand{\ka}{\kappa}

\newcommand{\ga}{\gamma}

\def\bea{\begin{eqnarray}\new\begin{array}{cc}}
\def\ee{\end{array}\end{eqnarray}}

\newcommand{\beq}[1]{\begin{equation}\label{#1}}
\newcommand{\eq}{\end{equation}}
\newcommand{\beqn}[1]{\begin{small} \begin{eqnarray}\label{#1}}
\newcommand{\eqn}{\end{eqnarray} \end{small}}
\newcommand{\p}{\partial}
\def\sq2{\sqrt{2}}

\newcommand{\GLN}{{\rm GL}(N,{\mathbb C})}

\def\sl2{{\rm sl}(2, {\mathbb C})}

\def\f1#1{\frac{1}{#1}}

\newcommand{\bz}{\bar{z}}

\newcommand{\bA}{\bar{A}}

\def\mC{{\mathbb C}}
\def\mZ{{\mathbb Z}}
\def\mR{{\mathbb R}}

\def\frak{\mathfrak}

\def\gg{{\frak g}}

\newcommand{\Mat}{ {\rm Mat}(N,\mathbb C) }
\def\clA{\mathcal{A}}

\def\clC{\mathcal{C}}
\def\clD{\mathcal{D}}

\def\clG{\mathcal{G}}

\def\clO{\mathcal{O}}
\def\clH{\mathcal{H}}
\def\clK{\mathcal{K}}

\def\clL{\mathcal{L}}
\def\clM{\mathcal{M}}

\def\clX{\mathcal{X}}
\def\clY{\mathcal{Y}}

\def\bag2{{\bf g_2}}
\def\bas8{{\bf so(8)}}

\def\sr2{\sqrt{2}}
\newcommand{\ran}{\rangle}
\newcommand{\lan}{\langle}
\def\f1#1{\frac{1}{#1}}

\newtheorem{predl}{Proposition}[section]

\newtheorem{theor}{Theorem}[section]
\newenvironment{proof}{\par\noindent{\bf Proof.}}{\hfill$\scriptstyle\blacksquare$}
\def\Ad{{\rm Ad}}

\def\res{\mathop{\hbox{Res}}\limits}
\newcommand{\mats}[4]{\left(\begin{array}{cc}{#1}&{#2}\\ {#3}&{#4}
\end{array}\right)}

\begin{document}

\setcounter{page}{1}

\begin{center}

\

\vspace{5mm}

{\LARGE{\bf Integrable deformations of principal chiral model from}}

\vspace{3mm}

{\LARGE{\bf solutions  of associative Yang-Baxter equation}}



 \vspace{20mm}

 {\Large {D. Domanevsky}}\footnote{Institute for Theoretical and Mathematical Physics, Lomonosov Moscow State University, Moscow, 119991, Russia;
 e-mail: danildom09@gmail.com.}
\quad
 {\Large {A. Levin}}\footnote{National Research University Higher School of Economics, Russian Federation,
 Usacheva str. 6,  Moscow, 119048, Russia;
 National Research Centre ''Kurchatov Institute'',
Academician Kurchatov square, 1, Moscow, 123182, Russia;
e-mail: alevin57@gmail.com.
 }
\quad
 {\Large {M. Olshanetsky}}\footnote{National Research Centre ''Kurchatov Institute'',
Academician Kurchatov square, 1, Moscow, 123182, Russia; Institute for Information Transmission Problems RAS (Kharkevich Institute),
Bolshoy Karetny per. 19, Moscow, 127994, Russia; e-mail: mikhail.olshanetsky@gmail.com.}
\quad
 {\Large {A. Zotov}}\footnote{Steklov Mathematical Institute of Russian
Academy of Sciences, Gubkina str. 8, 119991, Moscow, Russia; Institute for Theoretical and Mathematical Physics, Lomonosov Moscow State University, Moscow, 119991, Russia; National Research Centre ''Kurchatov Institute'',
Academician Kurchatov square, 1, Moscow, 123182, Russia; e-mail: zotov@mi-ras.ru.}

%
%
%
%
%
%



\end{center}

\vspace{10mm}

\begin{abstract}
We describe deformations of the classical principle chiral model and 1+1 Gaudin model
related to ${\rm GL}_N$ Lie group. The deformations are generated
by  $R$-matrices satisfying the associative Yang-Baxter equation. Using the coefficients
of the expansion for these $R$-matrices we
derive equations of motion based on a certain ansatz for $U$-$V$ pair satisfying the
Zakharov-Shabat equation. Another deformation comes from the twist function, which
we identify with the cocentral charge in the affine Higgs bundle underlying the Hitchin
approach to 2d integrable models.
\end{abstract}

\bigskip

\newpage
\bigskip
{\small{ \tableofcontents }}



\section{Introduction}\label{sec1}
\setcounter{equation}{0}

\paragraph{1+1 integrable field theories.} In this paper we consider two-dimensional (or, equivalently 2d or 1+1) integrable field theories of certain type.
Let us begin with two widely known examples. The first one is the Landau-Lifshitz model \cite{LL,Skl}. It is a model of an anisotropic ferromagnetic solid, described by equation
  \beq{q01}
  \begin{array}{l}
  \displaystyle{
 \p_t {\vec S}={\ti c}_1 {\vec S}\times J({\vec S})+{\ti c}_2 {\vec S}\times \p_x^2{\vec S}\,,
 \qquad J({\vec S})=(J_1S_1,J_2S_2,J_3S_3)\,,
 }
 \end{array}
 \eq
where $\vec{S}(t,x)=(S_1(t,x),S_2(t,x),S_3(t,x))\in\mC^3$ is a magnetization vector, ${\ti c}_1\,,{\ti c}_2$ and $J_1\,,J_2\,,J_3$ are some constants. The components of $\vec{S}(t,x)$ are fields on the two-dimensional space-time, $x$ is a space coordinate on a circle, and $t\in\mR$ is a time variable. The periodic boundary conditions
are assumed: ${\vec S}(t,x)={\vec S}(t,x+2\pi)$. We will deal with the
matrix form of the Landau-Lifshitz equation (\ref{q01}):
  \beq{q001}
  \begin{array}{l}
  \displaystyle{
 \p_t {S}={c}_1 [{S}, J({S})]+{c}_2 [{S}, \p_x^2{S}]\,,
 \qquad S=\sum\limits_{k=1}^3S_k\sigma_k\,,\quad J(S)=\sum\limits_{k=1}^3S_k J_k\sigma_k\,,
 }
 \end{array}
 \eq
 where $S$ is a traceless $2\times 2$ matrix, $S_k$ -- its components
 in the basis of the Pauli matrices $\sigma_k$ and $c_{1,2}=2\imath{\ti c}_{1,2}$.

The second example is the principal chiral model \cite{ZM,FR,Chered,FTbook}, which we write in the form (see \cite{Z11}):
\beq{q02}
\left\{ \begin{array}{l}
\p_tl_1-k\p_x l_0+\frac{2}{z_1-z_2}[l_1,l_0]=0\,,
\\ \ \\
\p_tl_0-k\p_xl_1=0\,,
\end{array}\right.
 \eq
where $l_0=l_0(t,x)\,,l_1=l_1(t,x)\in\Mat$ and $z_1,z_2,k\in\mC$. When $z_1=1$, $z_2=-1$
and $l_0=\p_t g g^{-1}$, $l_1=k\p_x g g^{-1}$ for some $g=g(t,x)\in\Mat$,
the upper equation become an identity, and the lower equation is the equation of motion of the
principal chiral model. We rewrite (\ref{q02}) in terms of
$S^1=\frac{1}{2}(l_0+l_1)$ and $S^2=\frac{1}{2}(l_0-l_1)$ as
\beq{q03}
\left\{\begin{array}{l}
\p_t S^1-k\p_x S^1=-\frac{2}{z_1-z_2}[S^1,S^2]\,,
\\ \ \\
\p_t S^2+k\p_x S^2=\frac{2}{z_1-z_2}[S^1,S^2]\,.
\end{array}\right. \eq

Both systems of equations (\ref{q01})--(\ref{q001}) and (\ref{q02})--(\ref{q03}) are represented in the
form of the Zakharov-Shabat equation
\cite{ZaSh,FTbook}:
  \beq{q04}
  \begin{array}{l}
  \displaystyle{
 \p_t U(z)-k\p_x V(z)=[U(z),V(z)]\,,
 }
 \end{array}
 \eq
 where $U(z),V(z)$ is a pair of matrix-valued functions of $t,x$,
  and $z\in\mC$ is the spectral parameter. The representation (\ref{q04})
  leads to integrability by means of the classical inverse scattering method.

\paragraph{Finite-dimensional mechanics.} In the limit when all the fields are independent of the space variable $x$, the
field theory equations of motion become ordinary differential equations describing some
finite-dimensional mechanics. This mechanics is integrable since the limit corresponds to $k=0$,
thus giving the Lax equations from (\ref{q04}):
  \beq{q05}
  \begin{array}{l}
  \displaystyle{
 {\dot L}(z)=[L(z),M(z)]\,,
 }
 \end{array}
 \eq
 where $L(z)$ is the limit of $U(z)$, and $M(z)$ is the limit of $V(z)$. In particular,
 the equation (\ref{q001}) in this limit takes the form (with $c_1=1$)
  \beq{q06}
  \begin{array}{l}
  \displaystyle{
  {\dot S}= [{S}, J({S})]\,,
 }
 \end{array}
 \eq
 which is an equation of motion for the (complexified) Euler top, where
 $S_1,S_2,S_3$ are angular momenta and $J_1$, $J_2$, $J_3$ are inverse components
 of the inertia tensor written in the principle axes. Similarly, the equations (\ref{q03})
 turn into
\beq{q07}
\begin{array}{l}
{\dot S}^1=-{\dot S}^2=-\frac{2}{z_1-z_2}[S^1,S^2]\,,
\end{array}
\eq
which is an equation for the rational Gaudin model with two marked points ${z_1,z_2}$ on a Riemann sphere.
The upper given Gaudin model can be generalized to trigonometric and elliptic cases with an arbitrary number of
marked points \cite{STS}. The Euler top (\ref{q06}) can be generalized to higher rank integrable cases.
Then the matrix $S\in\Mat$ is of an arbitrary size, and $J(S)$ is some linear map. Equation
(\ref{q06}) for $S\in\Mat$ is called the Euler-Arnold top on $\GLN$ Lie group.

\paragraph{Construction by $R$-matrices.} In this paper we use construction of integrable Euler-Arnold type tops obtained by a
special class of  quantum $R$-matrices.
In the general case a quantum (non-dynamical) $R$-matrix $R^{\hbar}_{12}(z_1,z_2)$ is a solution of the quantum Yang-Baxter equation (QYBE)
  \beq{q08}
  \begin{array}{l}
  \displaystyle{
 R^{\hbar}_{12}(z_1,z_2)  R^{\hbar}_{13}(z_1,z_3) R^{\hbar}_{23}(z_2,z_3) =
 R^{\hbar}_{23}(z_2,z_3)  R^{\hbar}_{13}(z_1,z_3)  R^{\hbar}_{12}(z_1,z_2)\,.
 }
 \end{array}
 \eq
 Being written in the fundamental representation of ${\rm GL}_N$ Lie group $R$-matrix become
 $\Mat^{\otimes 2}$-valued function of
 the Planck constant $\hbar$ and the spectral parameters $z_1,z_2$ (in fact, in our case $R^{\hbar}_{12}(z_1,z_2)=R^{\hbar}_{12}(z_1-z_2)$). That is
  \beq{q09}
  \begin{array}{l}
  \displaystyle{
  R^\hbar_{12}(z)=\sum\limits_{ijkl=1}^N R_{ij,kl}(\hbar,z)E_{ij}\otimes E_{kl}\in\Mat^{\otimes 2}\,,
   }
 \end{array}
 \eq
 where $E_{ij}$ are matrix units in $\Mat$ and $R_{ij,kl}(\hbar,z)$ is a set of functions.
 In the quantum Yang-Baxter equation (\ref{q08}) $R^{\hbar}_{ab}(z_a,z_b)=R^{\hbar}_{ab}(z_a-z_b)$ is understood as an element of $\Mat^{\otimes 3}$ acting non-trivially on the $a$-th and $b$-th tensor components only.
 We consider a class of quantum $R$-matrices satisfying also the
 associative Yang-Baxter equation \cite{FK}:
  \beq{q10}
  \begin{array}{c}
    \displaystyle{
  R^{\hbar}_{12} R^{\eta}_{23} = R^{\eta}_{13} R^{\hbar-\eta}_{12} + R^{\eta-\hbar}_{23} R^{\hbar}_{13}\,,
   \qquad R^u_{ab} = R^u_{ab}(z_a-z_b)\,.
 }
 \end{array}
 \eq
 In contrast to (\ref{q08}) the associative Yang-Baxter equation (AYBE) is a non-trivial functional equation
 even in the scalar case (when $N=1$):
  \beq{q11}
  \begin{array}{c}
  \displaystyle{
\phi(\hbar,z_{12})\phi(\eta,z_{23})=
\phi(\hbar-\eta,z_{12})\phi(\eta,z_{13})+\phi(\eta-\hbar,z_{23})\phi(\hbar,z_{13})\,,\quad z_{ab}=z_a-z_b\,.
 }
 \end{array}
 \eq
 In this case $R$-matrix become a function satisfying the relation
 (\ref{q11}) known as the genus one Fay identity. It has
 a solution given by the elliptic Kronecker function \cite{Weil}:
  \beq{q12}
  \begin{array}{l}
  \displaystyle{
 \phi(\hbar,z)=\frac{\vth'(0)\vth(\hbar+z)}{\vth(\hbar)\vth(z)}\,,
 \quad
 \vth(z)=-\sum _{k\in \mathbb Z} \exp
\left ( \pi \imath \tau (k+\frac{1}{2})^2 +2\pi \imath
(z+\frac{1}{2})(k+\frac{1}{2})\right )\,.
 }
 \end{array}
 \eq
where $\vth(z)$ is the first Jacobi theta-function. In the general case the QYBE (\ref{q08}) and the AYBE (\ref{q10})
have different (but intersecting) sets of solutions. At the same time, a certain class of solutions of AYBE satisfying
additional properties (see below the skew-symmetry (\ref{q34}) and the unitarity (\ref{q35})), satisfies also the QYBE.
In this paper we consider exactly these type solutions of AYBE, that is we consider a subset of solutions of QYBE satisfying
also AYBE. For this reason we refer to any solution of (\ref{q10}) as to $R$-matrix. This is important for our study since AYBE can be used for constructing different type integrable
structures (see below) likewise the Fay identity (\ref{q11}) is used for constructing the Lax equations and other
important objects.

Summarizing, the quantum $R$-matrix
satisfying AYBE (\ref{q10}) can be viewed as $\Mat^{\otimes 2}$ valued matrix generalization
of the Kronecker elliptic function (\ref{q12}), see \cite{Pol} and \cite{LOZ15}. This point of view leads to
a problem of description of non-commutative generalizations of the elliptic function identities \cite{Z16}.

AYBE can be used for constructing Lax pairs for classical integrable systems of the Euler-Arnold and Gaudin types \cite{LOZ14,LOZ16}. In particular, one may describe a family of the Euler-Arnold tops (\ref{q06})
with the inverse inertia tensor
 \beq{q121}
  \begin{array}{c}
  \displaystyle{
  J(S)=\frac{1}{N}\,\tr_2\Big(m_{12}(0){S}_2\Big)\,,
 }
 \end{array}
 \eq
where $\tr_2$ is a trace over the second tensor component, and $m_{12}(z)$ is the coefficient
appearing in the quasi-classical expansion of $R^\hbar_{12}(z)$ near $\hbar=0$:
  \beq{q13}
  \begin{array}{l}
  \displaystyle{
R^\hbar_{12}(z)=\frac{1}{\hbar}\,1_N\otimes 1_N+r_{12}(z)+\hbar\,m_{12}(z)+O(\hbar^2)\,.
 }
 \end{array}
 \eq
Here $r_{12}(z)$ is the classical $r$-matrix, and the Lax matrix is defined as
 \beq{q131}
  \begin{array}{c}
  \displaystyle{
  L(S,z)=\frac{1}{N}\,\tr_2\Big(r_{12}(z){S}_2\Big)\,.
 }
 \end{array}
 \eq
 A wide range of possible applications of AYBE to this type
 construction of integrable systems can be found in \cite{LOZ14,LOZ16,GSZ,SZ}.

\paragraph{1+1 field theories from finite-dimensional mechanics.}
In the above description
 we came to some simple finite-dimensional integrable systems by taking the limit $k=0$
in 2d integrable field theories. The inverse problem is to describe integrable
field theory generalization for a given finite-dimensional integrable model,
where integrability of 1+1 field theory is understood in the sense of the Zakharov-Shabat equation (\ref{q04})
and existence of the classical field $r$-matrix structure (of Maillet type), while the finite-dimensional integrability is in the Liouville theorem sense
through the Lax equation (\ref{q05}) and the classical $r$-matrix structure. In the general case this problem is too complicated. At the same time
certain families of finite-dimensional integrable systems can be extended to the 1+1 integrable field theories
using different approaches. One of them is to use reductions from integrable hierarchies of KP or 2d Toda types,
and treat certain time variable in the hierarchy as the space variable $x$. In this way
the field analogues of the Calogero-Moser and Ruijsenaars-Schneider models were proposed in \cite{Krich2,ZZ}.

Another approach is based on the affine Higgs bundles through extending the Hitchin approach to finite-dimensional integrable systems \cite{LOZ}. This construction provides
 1+1 field analogues of the spin Calogero-Moser model and its multipole generalizations. These models
 are non-ultralocal \cite{Z24}, while the models under consideration in this paper are ultralocal.
 These two type of models are related with each other through 1+1 version of the classical IRF-Vertex correspondence
 \cite{LOZ,Z24,AtZ1}. It is important to mention that the construction based on the affine Higgs bundles
 explains that $U$-matrix (in the ultralocal case) entering the Zakharov-Shabat equation (\ref{q04}) has the same form
 as the Lax matrix $L(z)$ from (\ref{q05}). Therefore, we conclude that
  \beq{q17}
  \begin{array}{l}
  \displaystyle{
U(S,z)=L(S,z)=\frac{1}{N}\,\tr_2\Big(r_{12}(z){S}_2\Big)\,.
 }
 \end{array}
 \eq

Based on the explicit form of $U$-matrix (\ref{q17}) one can deduce higher rank generalization of 1+1 Landau-Lifshitz model
(\ref{q001}) using a certain set of $R$-matrix identities \cite{AtZ2}:
 \beq{q851}
  \begin{array}{c}
  \displaystyle{
  \p_t S+\frac{k^2}{c}\,[S,\p^2_x S]+2[S,s_0 J(S)+k E(\p_x S)]=0\,,
 }
 \end{array}
 \eq
where $c,s_0$ are some constants, and $J(S)$, $E(S)$ are determined through the coefficients
of the expansion (\ref{q13}). Some details are given in the last Section.
That is, the third approach for
constructing 1+1 integrable field theories is to use $R$-matrix identities of AYBE type.

One more approach to  the 2d integrable systems is by  K. Costello and M. Yamazaki \cite{CY}. It is based on the 4d Chern-Simons theory \cite{Vicedo,Lacroix2}. It was demonstrated in
 \cite{LOZ22} that the 4d construction essentially coincides with the affine
Higgs bundles description.

\paragraph{Purpose of the paper.} Nowadays there is a growing interest to the studies of 1+1 integrable
field theories of Gaudin type \cite{Vicedo,LOZ22,Lacroix2}. These are the models, which $U$-matrix has more than one poles in spectral
parameter. An example is given by the principle chiral model (\ref{q02})-(\ref{q03}). Its $U$-matrix has a pair of simple poles (marked points).
In \cite{Chered} I. Cherednik constructed elliptic version of (\ref{q02})-(\ref{q03}) and described trigonometric and rational
degenerations including the 7-vertex deformation of the 6-vertex XXZ model. For an arbitrary number of
marked points the 1+1 elliptic Gaudin model was described in \cite{Z11}.

The problem is to describe possible integrable deformations of these models. In this paper we study two
types of deformations. The first one is given by the so-called twist function. This deformation
 corresponds to multiplication of
$U$-matrix by some function of spectral parameter. We will show that the twist function is naturally
identified with the cocentral charge in the affine Higgs bundles. Another set of deformations comes
from opportunity to use deformed $r$-matrices in the construction of $U$-matrix through (\ref{q17}).
In our description any $r$-matrix can be used, which comes in the quasi-classical limit (\ref{q13})
from a solution of the associative Yang-Baxter equation (\ref{q10}) with some special additional properties.
The multipole $U$-matrix then takes the form
  \beq{q177}
  \begin{array}{l}
  \displaystyle{
U(z)=\frac{1}{N}\sum\limits_{i=1}^n\tr_2\Big(r_{12}(z-z_i){S}^i_2\Big)\,.
 }
 \end{array}
 \eq
 The case of two marked points $n=2$ corresponds to the principle chiral model.

The paper is organized as follows. First, we recall the general construction underlying Hitchin approach to
2d integrable systems \cite{LOZ}. The Higgs field satisfies
the moment map equation, and its solution provides the $U$-matrix of some 2d integrable model. A new result
here is a description of possibility to make the cocentral charge to be dependent
on the spectral parameter.
In the construction of the affine Higgs bundles, we identify the twist function with the cocentral charge.
Next, we proceed to the associative Yang-Baxter equation. We describe a set of identities
underlying Lax equation and the Zakharov-Shabat equation. Then $R$-matrix identities are used to describe
1+1 Gaudin models based on $R$-matrices satisfying AYBE and certain additional properties.
Finally, we discuss deformations coming from introducing the
twist function. It is closely related to the Gaudin model since it adds new poles to $U$-matrix.
In the Appendix we give a very brief review of the Hitchin approach to integrable systems from
the group-theoretical viewpoint.


\paragraph{Acknowledgments.}
 The work of M.A. Olshanetsky was supported by the Russian Science
Foundation grant 24-12-00178. The work of Danil Domanevsky was supported by the Foundation
for the Advancement of Theoretical Physics and Mathematics ''BASIS''.


\section{General construction: Lax operator as connection in the affine Higgs bundle}
\setcounter{equation}{0}

In this paper we discuss a special (although a wide) class of integrable models described
by $U$-matrices, which can be considered as sections of bundles over elliptic curves (or their degenerations)
with simple poles (in spectral parameter $z$) at some set of points $z_1,...,z_n$.
These models can be described in the framework of
the Hitchin approach to integrable systems. Namely, following \cite{LOZ} we define the Lax operator
as a section of the affine Higgs bundle which turns into some $U$-matrix of an integrable model after
performing symplectic reduction with respect to gauge symmetries. That is, the reduced phase space
is the moduli space of the (affine) Higgs bundles. A detailed description of this approach is too long.
For this reason we give a brief review of the construction in the Appendix.
Main result of the symplectic reduction is that an $U$-matrix satisfies the moment map equation.
It can be solved explicitly for some concrete examples of bundles over elliptic curves.

In the finite-dimensional mechanics the Hitchin approach works as follows.
The initial phase space is infinite-dimensional. It is a space of pairs $(\Phi,\bar A)$
including the Higgs field $\Phi$ and the component $\bar A$ of the connection defining the
holomorphic structure, see (\ref{hb1}). This space is
equipped with the canonical symplectic structure (\ref{cf}). By performing the symplectic
reduction with respect to the action of gauge group, we arrive to the moment map equation
  \beq{q30101}
  \begin{array}{c}
    \displaystyle{
\p_{\bz}\Phi+[\bA,\Phi]=\sum_{a=1}^n S^a\de(z-z_a,\bz-\bz_a)\,.
 }
 \end{array}
 \eq
 A gauge fixation condition makes $\bar A$ be either diagonal (an element of the Cartan subalgebra in the Lie algebra
 of the corresponding Lie group) or equal to zero.
 Then solutions of (\ref{q30101}) with respect to $\Phi$ provide Lax matrices of different
 finite-dimensional integrable systems of classical mechanics.
 Put it differently, as a result of the reduction $\Phi$ turns into some Lax matrix, see \cite{LOZ,SmZ}.


\subsection{$U$-matrices and integrability}
Now, let us proceed to the field case. The moment
map equation takes the form\footnote{In the Appendix (\ref{q301})
appears as $\mu_0=0$ in (\ref{me}) with $\mu_0$ (\ref{mom1}).}:
%
  \beq{q301}
  \begin{array}{c}
    \displaystyle{
\p_{\bz}\Phi- {k}\p_x \bA+[\bA,\Phi]=\sum_{a=1}^n S^a\de(z-z_a,\bz-\bz_a)\,.
 }
 \end{array}
 \eq
 Equation (\ref{q301}) is the field generalization of the moment map equation (\ref{q30101})
 Here we study the case $\Sigma$ is an elliptic curve and
  $G=\GLN$, so that $\Phi$, $\bar A$ and the set of  $S^a$, $a=1,...,n$
 become matrices from $\Mat$.
 Notice that using the gauge group transformations one can make ${\bar A}$ be diagonal
  (in the general case ${\bar A}$ is transformed to an element from the Cartan subalgebra of $\gg$).
 Perform a conjugation of both sides with $\exp((z-{\bar z}){\bar A})$.
 Then we get (\ref{q301}) in the ''holomorphic gauge'':
  \beq{q302}
  \begin{array}{c}
    \displaystyle{
\p_{\bz}\ti\Phi- {k}\p_x \bA=\sum_{a=1}^n S^a\de(z-z_a,\bz-\bz_a)\in\Mat\,,
 }
 \end{array}
 \eq
where $\ti\Phi=e^{-(z-{\bar z}){\bar A}}\Phi e^{(z-{\bar z}){\bar A}}$.
 Let ${k}$ be a constant for simplicity. More complicated cases are considered below in the next subsection.
Solution of this equation requires some boundary conditions and a gauge fixation of ${\bar A}$.
The gauge fixation of ${\bar A}$ depends on topological properties of the Higgs bundle. In fact,
the Higgs field $\ti\Phi$ is a section of the affine Higgs bundle ${\rm End}(E^{aff})$, see Appendix.
 The mentioned above topological property is the degree of the underlying bundle $E$.
When ${\rm deg}(E)=0$ the matrix ${\bar A}$ can be transformed to the constant diagonal matrix. The diagonal elements
become dynamical variables. Then the solution of (\ref{q302}) yields the $U$-matrix
of the field analogue of the Calogero-Moser model or its spin and multispin (multipole) generalizations \cite{LOZ,Z24}.

In this paper we study the case corresponding to ${\rm deg}(E)=1$. Then ${\bar A}$ can be gauged to zero,
and (\ref{q302}) takes the form (when ${\bar A}\rightarrow 0$, the Higgs field $\ti\Phi\rightarrow U(z)$ becomes
the $U$-matrix)
  \beq{q303}
  \begin{array}{c}
    \displaystyle{
\p_{\bz}U(z)=\sum_{a=1}^n S^a\de(z-z_a,\bz-\bz_a)\in\Mat\,,
 }
 \end{array}
 \eq
 that is $\ti\Phi$ has $n$ simple poles at $z_1,...,z_n$ with residues
  \beq{q304}
  \begin{array}{c}
    \displaystyle{
\res\limits_{z=z_a}U(z)=S^a\in\Mat\,.
 }
 \end{array}
 \eq
Explicit expression for $U$-matrix in elliptic functions will be given in the
next Section, see (\ref{q482}). Here we emphasize that the equation (\ref{q303}) or (\ref{q304})
is free of derivatives with respect to the loop variable $x$. That is
(\ref{q303}) has {\em the same form as in the finite-dimensional case}.
Therefore, its solution has the same form as in mechanics but the residues $S^a$ are now
considered as functions of $x$.

Integrability of 2d models under consideration is achieved through the classical $r$-matrix
structure. The Poisson brackets are as follows:
  \beq{q305}
  \begin{array}{c}
  \displaystyle{
   \{S^a_{ij}(x),S^b_{kl}(x')\}=\delta^{ab}\Big(-S^a_{il}(x)\delta_{kj}+S^a_{kj}(x)\delta_{il}\Big)\delta(x-x')
 }
 \end{array}
 \eq
for $i,j,k,l=1,...,N$ and $a,b=1,...,n$.
And the classical $r$-matrix structure is similar to the finite-dimensional case:
  \beq{q306}
  \begin{array}{c}
  \displaystyle{
\{U_1(z,x),U_2(w,x')\}=
\Big(
 [U_1(z,x),{\bf r}_{12}(z,w|x)]-[U_2(w,x'),{\bf r}_{21}(w,z|x)]\Big)\delta(x-x')\,.
  }
 \end{array}
 \eq
 Moreover, in our case the classical $r$-matrix is skew-symmetric (${\bf r}_{21}(w,z|x)=-{\bf r}_{12}(z,w|x)$),
 depends on the difference of spectral parameters ${\bf r}_{12}(z,w|x)={\bf r}_{12}(z-w|x)$  and
 it is non-dynamical: ${\bf r}_{12}(z-w|x)=r_{12}(z-w)$.
 This type $r$-matrix satisfies the standard classical Yang-Baxter equation for non-dynamical skew-symmetric $r$-matrices
  \beq{q307}
  \begin{array}{c}
  \displaystyle{
[r_{12}(z_1-z_2),r_{13}(z_1-z_3)]+
[r_{12}(z_1-z_2),r_{23}(z_2-z_3)]+
[r_{13}(z_1-z_3),r_{23}(z_2-z_3)]=0\,.
  }
 \end{array}
 \eq
 %

 Introduce the matrix
  \beq{q308}
  \begin{array}{c}
  \displaystyle{
U(z)=\frac{1}{N}\sum\limits_{a=1}^n\tr_2\Big(r_{12}(z-z_a)S^a_2\Big)\in\Mat\,,
  }
 \end{array}
 \eq
 where $\tr_2$ is the trace over the second tensor component\footnote{We explain this notation in the next Section in detail.}.
 Let the Poisson structure is given by (\ref{q305}) and the $r$-matrix entering (\ref{q308})
 is skew-symmetric and satisfies the classical Yang-Baxter equation (\ref{q307}).
 Then the following statement is easily verified.
 \begin{predl}
 The matrix $U(z)$ (\ref{q308}) satisfies the classical $r$-matrix structure (\ref{q306}).
 \end{predl}

 This Proposition guarantees integrability in the following sense. Existence of the classical $r$-matrix
 structure (\ref{q306}) is a sufficient condition for the Poisson commutativity
 $\{\tr(T^k(z,2\pi)),\tr(T^m(w,2\pi))\}=0$ of the traces of powers of the monodromy matrices
 \beq{q309}
 \begin{array}{c}
  \displaystyle{
 T(z,x)={\rm Pexp}\Big( \frac{1}{k}\int\limits_0^x dx'\,  U(z,x') \Big)\,.
  }
 \end{array}
\eq

\subsection{Cocentral charge as twist function}

In the above analysis of $U$-matrices we assumed $k$ be a constant. At the same time it satisfies
the equation 2 from (\ref{me})
 \beq{q310}
 \begin{array}{c}
  \displaystyle{
 \p_{\bar z}k(z)=\sum\limits_{b=1}^m s_b\delta(z-w_b,{\bar z}-{\bar w}_b)\,,
  }
 \end{array}
\eq
 i.e.
 \beq{q311}
 \begin{array}{c}
  \displaystyle{
 \res\limits_{z=w_b}k(z)=s_b\,,\quad b=1,...,m,
  }
 \end{array}
\eq
 so that constant $k$ is the simplest case corresponding to $m=0$ or $s_b=0$ for all $b$.

It is important to notice that the positions of simple poles $w_b$ may do not coincide with
the positions of marked points $z_a$ in (\ref{q304}).

Consider two examples.

\emph{1. $\Si$ is a rational curve.} By representing $k(z)$ in the form
 \beq{q312}
 \begin{array}{c}
  \displaystyle{
 {k}(z)=\prod_{b=1}^m\frac{z-y_b}{z-w_b},
  }
 \end{array}
\eq
we then obtain (\ref{q312}) as a solution of (\ref{q310}) with
\beq{nr}
s_b=\frac{\prod\limits_{c=1}^m(w_b-y_c)}{\prod\limits_{c:c\neq b}^m(w_b-w_c)}\,.
\eq
This solution satisfies (\ref{v3}).


\emph{2. $\Si$ is an elliptic curve.}
In this case solutions of the moment equation (\ref{v2}) should be double-periodic.
We take
\beq{nr1}
{k}(z)=\prod_{b=1}^m\frac{\vth(z-y_b)}{\vth(z-w_b)}\,.
\eq
For the double-periodicity positions of poles and zeros should satisfy the condition
$$
\sum_{b=1}^my_b=\sum_{b=1}^mw_b\,.
$$
In this case we have solution (\ref{nr1}) of (\ref{q310}) with
$$
s_b=\frac{\prod\limits_{c=1}^m\vth(w_b-y_c)}{\vth'(0)\prod\limits_{c:c\neq b}^m\vth(w_b-w_c)}\,.
$$

The transition from the
constant $k$ to $k(z)$ means transition from the connection $k\p_x-U(z)$ to $k(z)\p_x-U(z)$.
From the point of view of the Zakharov-Shabat equation (\ref{q04})
 it is equivalent to the transition from $U(z)$ in
  \beq{q313}
  \begin{array}{l}
  \displaystyle{
 \p_t U(z)-k(z)\p_x V(z)=[U(z),V(z)]\,,
 }
 \end{array}
 \eq
 to ${\ti U}(z)$:
 \beq{q314}
 \begin{array}{c}
  \displaystyle{
 U(z) \rightarrow {\ti U}(z)=\frac{k}{k(z)}\,U(z)\,.
  }
 \end{array}
\eq
 in
  \beq{q315}
  \begin{array}{l}
  \displaystyle{
 \p_t {\ti U}(z)-k\p_x V(z)=[{\ti U}(z),V(z)]\,.
 }
 \end{array}
 \eq
Then from (\ref{q306}) we conclude that
  \beq{q316}
  \begin{array}{c}
  \displaystyle{
\{{\ti U}_1(z,x),{\ti U}_2(w,x')\}=
\Big(
 [{\ti U}_1(z,x),{\ti r}_{12}(z,w)]-[{\ti U}_2(w,x'),{\ti r}_{21}(w,z)]\Big)\delta(x-x')\,,
  }
 \end{array}
 \eq
 where
  \beq{q317}
  \begin{array}{c}
  \displaystyle{
{\ti r}_{12}(z,w)=\frac{k(z)}{k}\,r_{12}(z-w)\,.
  }
 \end{array}
 \eq
 Such redefinition is widely known in recent studies of 2d field theories. The function
 $k(z)/k$ is called the twist function (or, more precisely, its inverse $k/k(z)$).
 See \cite{Lacroix2,Vicedo} and references therein.

 Finally, we conclude that non-trivial dependence $k(z)$ provides additional poles in $U$-matrix through
 (\ref{q314}). We will consider this possibility in the end of the paper.


\section{Associative Yang-Baxter equation and finite-dimensional models}
\setcounter{equation}{0}

\subsection{$R$-matrices}



In this paper we consider the associative Yang-Baxter equation
  \beq{q31}
  \begin{array}{c}
    \displaystyle{
  R^{\hbar}_{12}(z_1-z_2) R^{\eta}_{23}(z_2-z_3) =
   R^{\eta}_{13}(z_1-z_3) R^{\hbar-\eta}_{12}(z_1-z_2) +
    R^{\eta-\hbar}_{23}(z_2-z_3) R^{\hbar}_{13}(z_1-z_3)
 }
 \end{array}
 \eq
written for $R$-matrices
in the fundamental representation of ${\GLN}$ Lie group. That is $R_{12}^\hbar(z)\in\Mat^{\otimes 2}$
 is a matrix valued function. The indices $ab$ in $R_{ab}^\hbar(z)$ mean the numbers of tensor components, where
 it acts non-trivially. For example,
  \beq{q32}
  \begin{array}{c}
    \displaystyle{
 R^\hbar_{13}(z)=\sum\limits_{ijkl=1}^N R_{ij,kl}(\hbar,z)E_{ij}\otimes 1_N\otimes E_{kl}\,.
 }
 \end{array}
 \eq
 Change of order of these indices means the permutation of the tensor components, i.e. compared to (\ref{q09})
 the expression for $R_{21}^\hbar(z)$ has the form
  \beq{q33}
  \begin{array}{l}
  \displaystyle{
  R^\hbar_{21}(z)=\sum\limits_{ijkl=1}^N R_{ij,kl}(\hbar,z)E_{kl}\otimes E_{ij}\,.
   }
 \end{array}
 \eq
Besides AYBE (\ref{q31}) $R$-matrices under consideration are assumed to satisfy two more properties.
The first one is skew-symmetry
  \beq{q34}
    \displaystyle{
  R^\hbar_{12}(z)=-R^{-\hbar}_{21}(-z)\,,
  }
  \eq
and the second is unitarity
\beq{q35}
\begin{array}{c}
    R^{\hbar}_{12} (z) R^{\hbar}_{21} (-z)
    =F^\hbar(z)\, 1_N \otimes 1_N\,,
\end{array}
\eq
where $F^\hbar(z)$ is a function depending on a choice of normalization.
For instance, it can be equal to $N^2\phi(N\hbar,z)\phi(N\hbar,-z)$ or $\phi(\hbar,z)\phi(\hbar,-z)$
in different cases.
The function $\phi$ entering
(\ref{q35}) through $F^\hbar(z)$ was defined in (\ref{q12}). That expression is valid for elliptic $R$-matrix, while
for trigonometric and rational $R$-matrices one should replace it with the corresponding degenerations of
(\ref{q12}):
  \beq{q36}
  \begin{array}{l}
  \displaystyle{
 \phi(\hbar,z)=\frac{\vth'(0)\vth(\hbar+z)}{\vth(\hbar)\vth(z)}
 \quad\stackrel{\rm trig.}{\longrightarrow}\quad
 \pi\frac{\sin(\pi(\hbar+z))}{\sin(\pi\hbar)\sin(\pi z)}
 \quad\stackrel{\rm rat.}{\longrightarrow}\quad
 \frac{\hbar + z}{\hbar z}\,.
 }
 \end{array}
 \eq
We also imply certain local behaviour of $R$-matrices near $z=0$ and $\hbar=0$.
The quasi-classical
limit (\ref{q13}) means that
  \beq{q361}
  \begin{array}{l}
  \displaystyle{
\res\limits_{\hbar=0}R^\hbar_{12}(z)=1_N\otimes 1_N\,.
 }
 \end{array}
 \eq
The local behaviour near $z=0$ is as follows:
  \beq{q362}
  \begin{array}{l}
  \displaystyle{
\res\limits_{z=0}R^\hbar_{12}(z)=\res\limits_{z=0}r_{12}(z)=NP_{12}=N\sum\limits_{ijkl=1}^NE_{ij}\otimes E_{kl}\,,
 }
 \end{array}
 \eq
where $P_{12}$ is the matrix permutation operator (it permutes a pair of vectors in tensor product).
In particular, it means that we have the following expansion near $z=0$:
%
  \beq{q363}
  \begin{array}{l}
  \displaystyle{
r_{12}(z) = \frac{1}{z}\, NP_{12} + r^{(0)}_{12} +  O(z)\,.
 }
 \end{array}
 \eq
From the skew-symmetry (\ref{q34}) one can easily find (skew)symmetric properties of
the coefficients of expansions (\ref{q13}) and (\ref{q363}):
  \beq{q364}
  \begin{array}{l}
  \displaystyle{
r_{12}(z) = -r_{21}(-z)\,,\quad
m_{12}(z) = m_{21}(-z)\,,\quad
r^{(0)}_{12}=-r^{(0)}_{21}\,,\quad
m_{12}(0) = m_{21}(0)
\,.
 }
 \end{array}
 \eq

It can be shown (see e.g. \cite{LOZ14}) that any solution of (\ref{q31}) obeying the properties
(\ref{q34}) and (\ref{q35}) satisfies also the quantum Yang-Baxter equation (\ref{q08}).
This is why we indeed deal with quantum $R$-matrices. More precisely, we deal with a subset
of solutions of the quantum Yang-Baxter equation (\ref{q08}), which also satisfy the AYBE (\ref{q31}).
Let us briefly describe families of $R$-matrices which satisfy all required properties.
Before considering non-trivial examples let us write down the most simple one, which
is the (properly normalized) rational Yang's $R$-matrix:
 \beq{q365}
 \begin{array}{c}
  \displaystyle{
R_{12}^{{\rm Yang},\hbar}(z)=\frac{1_N\otimes 1_N}{\hbar}+\frac{P_{12}}{z}\,.
  }
 \end{array}
\eq
For $N=1$ it coincides with the rational function $\phi(\hbar,z)$ (\ref{q36}).

\paragraph{Elliptic $R$-matrix.}
For definition of the Baxter-Belavin elliptic $R$-matrix \cite{BB}
the special matrix basis in $\Mat$ should be used:
  \beq{q37}
  \begin{array}{l}
  \displaystyle{
T_a=T_{a_1 a_2}=\exp\left(\frac{\pi\imath}{{ N}}\,a_1
 a_2\right)Q_1^{a_1}Q_2^{a_2}\,,
 \qquad
 a=(a_1,a_2)\in\mZ_{ N}\times\mZ_{ N}
 }
 \\ \ \\
  \displaystyle{
(Q_1)_{kl}=\delta_{kl}\exp(\frac{2\pi
 \imath}{{ N}}k)\,,
 \qquad
 (Q_2)_{kl}=\delta_{k-l+1=0\,{\hbox{\tiny{mod}}}\,
 { N}}\,,\quad k,l=1,..,N\,.
 }
 \end{array}
 \eq
In particular, $T_{(0,0)}=1_N$. The basis has the property $\tr(T_\al T_\be)=N\delta_{\al+\be,(0,0)}$.
 See details in \cite{BB} (and also Appendix in \cite{ZZ}). Then quantum elliptic $R$-matrix is as follows:
 \beq{q38}
 \begin{array}{c}
  \displaystyle{
 R_{12}^\hbar(z)=\sum\limits_{a\in\,\mZ_{ N}\times\mZ_{ N}} T_a\otimes T_{-a}
 \exp (2\pi\imath\,\frac{a_2z}{N})\,\phi(z,\frac{a_1+a_2\tau}{N}+\hbar)\in{\rm Mat}(N,\mC)^{\otimes 2}\,.
 }
 \end{array}
 \eq
Its classical limit yields
 \beq{q39}
 \begin{array}{c}
  \displaystyle{
 r_{12}(z)=E_1(z)1_N\otimes 1_N+\sum\limits_{a\neq(0,0)} T_a\otimes T_{-a}
 \exp (2\pi\imath\,\frac{a_2z}{N})\,\phi(z,\frac{a_1+a_2\tau}{N})
 }
 \end{array}
 \eq
with $E_1(z)=\p_z\log\vth(z)$ and
 \beq{q40}
 \begin{array}{c}
  \displaystyle{
 m_{12}(z)=\rho(z)1_N\otimes 1_N+\sum\limits_{a\neq(0,0)} T_a\otimes T_{-a}
 \exp (2\pi\imath\,\frac{a_2z}{N})\,f(z,\frac{a_1+a_2\tau}{N})\,,
 }
 \end{array}
 \eq
 where $f(z,u)=\p_w\phi(z,w)|_{w=u}$ and $\rho(z)=(E_1^2(z)-\wp(z))/2$. Then one finds
 \beq{q41}
 \begin{array}{c}
  \displaystyle{
 r_{12}^{(0)}=\!\sum\limits_{a\neq(0,0)}\! T_a\otimes T_{-a}\Big(2\pi\imath\,\frac{a_2}{N}+E_1(\frac{a_1\!+\!a_2\tau}{N}) \Big)\,,
 }
 \end{array}
 \eq
 \beq{q42}
 \begin{array}{c}
  \displaystyle{
 m_{12}(0)=\frac{\vth'''(0)}{3\vth'(0)}1_N\otimes 1_N-\sum\limits_{a\neq(0,0)} T_a\otimes T_{-a}
 E_2(\frac{a_1+a_2\tau}{N})
 }
 \end{array}
 \eq
 with $E_2(x)=-E_1'(z)=-\p_z^2\log\vth(z)$. With these results $J(S)$ from (\ref{q121}) takes the form:
 \beq{q43}
 \begin{array}{c}
  \displaystyle{
 J(S)=\frac{\vth'''(0)}{3\vth'(0)}1_N S_{0,0}-\sum\limits_{a\neq(0,0)} T_a S_a
 E_2(\frac{a_1+a_2\tau}{N})=\frac{\vth'''(0)}{3\vth'(0)}\,S-\sum\limits_{a\neq(0,0)} T_a S_a
 \wp(\frac{a_1+a_2\tau}{N})\,.
 }
 \end{array}
 \eq
 where we used relation the $E_2(z)=\wp(z)-\vth'''(0)/(3\vth'(0))$, $\wp(z)$ -- is the Weierstrass $\wp$-function.

 When $N=2$ case we obtain the 8-vertex Baxter's $R$-matrix. In this case
 \beq{q44}
 \begin{array}{c}
  \displaystyle{
 Q_1=\mats{-1}{0}{0}{1}\,,\quad Q_2=\mats{0}{1}{1}{0}
  }
 \end{array}
 \eq
and the basis matrices are $T_{00}=1_2=\sigma_0$, $T_{10}=-\sigma_3$, $T_{01}=\sigma_1$ and
$T_{11}=\sigma_2$, where $\{\sigma_a\}$, $a=0,1,2,3,4$ are the Pauli matrices:
 \beq{q45}
 \begin{array}{c}
  \displaystyle{
 \sigma_0=\mats{1}{0}{0}{1}\,,\quad
 \sigma_1=\mats{0}{1}{1}{0}\,,\quad
 \sigma_2=\mats{0}{-\imath}{\imath}{0}\,,\quad
 \sigma_3=\mats{1}{0}{0}{-1}\,.
  }
 \end{array}
 \eq
The $R$-matrix (\ref{q38}) for $N=2$ has the form
 \beq{q46}
 \begin{array}{c}
  \displaystyle{
 R_{12}^\hbar(z)
 =\vf_{00}\,\sigma_0\otimes\sigma_0
 +\vf_{01}\,\sigma_1\otimes\sigma_1
  +\vf_{11}\,\sigma_2\otimes\sigma_2
  +\vf_{10}\,\sigma_3\otimes\sigma_3\,,
  }
 \end{array}
\eq
 \beq{q47}
 \begin{array}{c}
  \displaystyle{
 \vf_{00}=\phi(z,\hbar)\,,\quad
 \vf_{10}=\phi(z,\hbar)\,,\quad
 \vf_{01}=e^{\pi\imath z}\phi(z,\hbar)\,,\quad
 \vf_{11}=e^{\pi\imath z}\phi(z,\hbar)\,.
 }
  \end{array}
 \eq
 In the form of $4\times 4$ matrix it is as follows:
 \beq{q48}
 \begin{array}{c}
 R_{12}^\hbar(z)=\left(
 \begin{array}{cccc}
 \vf_{00}+\vf_{10} & 0 & 0 & \vf_{01}-\vf_{11}
 \\
 0 & \vf_{00}-\vf_{10} & \vf_{01}+\vf_{11} & 0
  \\
 0 & \vf_{01}+\vf_{11} & \vf_{00}-\vf_{10} & 0
 \\
 \vf_{01}-\vf_{11} & 0 & 0 & \vf_{00}+\vf_{10}
 \end{array}
 \right)\,.
  \end{array}
 \eq
\paragraph{Elliptic $U$-matrix.} The equation (\ref{q303}) supplied with
the boundary conditions (quasi-periodic behaviour on the lattice of periods of elliptic curve)
 \beq{q481}
 \begin{array}{c}
  \displaystyle{
 U(z+1)=Q_1^{-1}U(z)Q_1\,,\qquad
 U(z+\tau)=Q_2^{-1}U(z)Q_2
 }
 \end{array}
 \eq
has the following solution:
 \beq{q482}
 \begin{array}{c}
  \displaystyle{
 U(z)=1_N\sum\limits_{j=1}^n S^j_{(0,0)}E_1(z-z_i)+\sum\limits_{j=1}^n\sum\limits_{a\neq(0,0)} S^j_a T_a
 \exp (2\pi\imath\,\frac{a_2(z-z_j)}{N})\,\phi(z-z_j,\frac{a_1+a_2\tau}{N})\,,
 }
 \end{array}
 \eq
where $\sum\limits_{j=1}^n S^j_{(0,0)}=0$. The first term is proportional to the identity matrix. It is not necessary and can be removed. Expression (\ref{q482}) coincides with (\ref{q308}), where $r$-matrix is (\ref{q39}).
On the other hand it coincides with the expression for the Lax matrix of elliptic Gaudin model \cite{STS}
or elliptic Schlesinger system \cite{CLOZ}.

\paragraph{Trigonometric $R$-matrices.} Let us begin with $N=2$ case.
In \cite{Chered} the following 7-vertex $R$-matrix was found:
  \beq{q49}
   \begin{array}{c}
  R^\hbar(z)=\left(\begin{array}{cccc} \coth(z)+\coth(\hbar) & 0 & 0 & 0\vphantom{\Big|}
  \\ 0 & \sinh^{-1}(\hbar) & \sinh^{-1}(z) & 0\vphantom{\Big|}
  \\ 0 & \sinh^{-1}(z) & \sinh^{-1}(\hbar) & 0\vphantom{\Big|}
  \\ -4\,e^{-2\Lambda}\sinh(z+\hbar) & 0 & 0 & \coth(z)+\coth(\hbar)          \vphantom{\Big|} \end{array} \right)\,,
  \end{array}
  \eq
 where $\Lambda$ is an arbitrary constant. In the limit $\Lambda\to +\infty$
 we come to the standard 6-vertex XXZ $R$-matrix.
  Classification of trigonometric solutions of AYBE (for $\GLN$ case) was suggested in
  \cite{Pol2}, see also brief review in \cite{LOZ16}.

\paragraph{Rational $R$-matrices.} Consider $N=2$ case. The simplest example
is the Yang's $R$-matrix (\ref{q365}), which in $N=2$ case becomes
the XXX 6-vertex $R$-matrix
 \beq{q50}
 \begin{array}{c}
  \displaystyle{
 R_{12}^{\hbar}(z)=
  \left(
  \begin{array}{cccc}
  1/\hbar+1/z & 0 & 0 & 0
  \\
  0 & 1/\hbar & 1/z & 0
  \\
   0 & 1/z & 1/\hbar & 0
   \\
   0 & 0 & 0 & 1/\hbar+1/z
   \end{array}
   \right)\,.
  }
 \end{array}
\eq
It has the following
 11-vertex deformation \cite{Chered}:
 \beq{q51}
 \begin{array}{c}
  \displaystyle{
 R_{12}^{{\rm 11v},\hbar}(z)=
  \left(
  \begin{array}{cccc}
  1/\hbar+1/z & 0 & 0 & 0
  \\
  -z-\hbar & 1/\hbar & 1/z & 0
  \\
   -z-\hbar & 1/z & 1/\hbar & 0
   \\
   -z^3-\hbar^3-2z^2\hbar-2z\hbar^2 & z+\hbar & z+\hbar & 1/\hbar+1/z
   \end{array}
   \right)\,.
  }
 \end{array}
\eq
This is a deformation of the Yang's $R$-matrix in the following sense:
 \beq{q52}
 \begin{array}{c}
  \displaystyle{
 \lim\limits_{\epsilon\rightarrow 0}\epsilon R_{12}^{{\rm 11v},\hbar\epsilon}(z\epsilon)=R_{12}^{{\rm Yang},\hbar}(z)\,.
  }
 \end{array}
\eq
Higher rank analogues of the 11-vertex $R$-matrix satisfying AYBE can be found in \cite{AtZ3}.

\subsection{$R$-matrix identities}
It was mentioned in the Introduction that the  AYBE (\ref{q31}) is a matrix generalization
of the addition formula (\ref{q11}) for the elliptic Kronecker function (\ref{q12}).
Therefore, using AYBE, the properties (\ref{q34})-(\ref{q35}) and (\ref{q361})-(\ref{q362})
and the expansions (\ref{q13}), (\ref{q363}) one can derive a wide set of identities likewise
it works in the scalar case for elliptic functions. For example,
 \beq{q53}
 \begin{array}{c}
  \displaystyle{
 \Big(r_{12}(z_1-z_2)+r_{23}(z_2-z_3)+r_{31}(z_3-z_1)\Big)^2=
  }
  \\ \ \\
    \displaystyle{
 =1_{N^3}N^2\Big( \wp(z_1-z_2)+\wp(z_2-z_3)+\wp(z_3-z_1) \Big)
  }
 \end{array}
\eq
is a matrix analogue of the scalar identity
 \beq{q54}
 \begin{array}{c}
  \displaystyle{
 \Big(E_1(z_1-z_2)+E_1(z_2-z_3)+E_1(z_3-z_1)\Big)^2=
 \Big( \wp(z_1-z_2)+\wp(z_2-z_3)+\wp(z_3-z_1) \Big)\,.
  }
 \end{array}
\eq
While $R_{12}^\hbar(z)$ is a matrix analogue of the function $\phi(\hbar,z)$,
the classical $r$-matrix is a matrix analogue of $E_1(z)=\p_z\log\vth(z)$.
Different non-trivial examples of the $R$-matrix identities can be found in
\cite{LOZ14,LOZ15,LOZ16,Z16}.

Here we list some identities required for the Lax equation and the Zakharov-Shabat equation.
The first one identity is
\beq{q55}
\begin{split}
    & \left[m_{13}\left(z_1-z_3\right), r_{12}\left(z_1-z_2\right)\right] = \left[r_{12}\left(z_1-z_2\right), m_{23}\left(z_2-z_3\right)\right] +\\ & + \left[m_{12}\left(z_1-z_2\right), r_{23}\left(z_2-z_3\right)\right] + \left[m_{13}\left(z_1-z_3\right), r_{23}\left(z_2-z_3\right)\right]\,.
    \end{split}
\eq
In particular,
\beq{q56}
\begin{split}
    & \left[m_{13}\left(0\right), r_{12}\left(z_a-z_b\right)\right] = \left[r_{12}\left(z_a-z_b\right), m_{23}\left(z_b-z_a\right)\right] +\\ & + \left[m_{12}\left(z_a-z_b\right), r_{23}\left(z_b-z_a\right)\right] + \left[m_{13}\left(0\right), r_{23}\left(z_b-z_a\right)\right]\,.
    \end{split}
\eq
In another limiting case (when $z_2\rightarrow z_1$) (\ref{q55}) yields
 \beq{q57}
 \begin{array}{c}
  \displaystyle{
 [m_{13}(z),r_{12}(z)]=[r_{12}(z),m_{23}(0)]-[\p_z m_{12}(z),NP_{23}]
  +[m_{12}(z),r_{23}^{(0)}]+[m_{13}(z),r_{23}^{(0)}]\,.
 }
 \end{array}
 \eq

Another identity is
 \beq{q58}
  \begin{array}{c}
  \displaystyle{
  r_{12}(z)r_{13}(z\!+\!w)-r_{23}(w)r_{12}(z)+r_{13}(z\!+\!w)r_{23}(w)=m_{12}(z)+m_{23}(w)+m_{13}(z\!+\!w)\,.
 }
 \end{array}
 \eq
In the limit $w\rightarrow 0$ it gives
 \beq{q59}
  \begin{array}{c}
  \displaystyle{
  r_{12}(z)r_{13}(z)=r_{23}^{(0)}r_{12}(z)-r_{13}(z)r_{23}^{(0)}
  -N\p_z r_{13}(z)P_{23}+m_{12}(z)+m_{23}(0)+m_{13}(z)\,.
 }
 \end{array}
 \eq

\subsection{Euler-Arnold tops, Gaudin models and Painlev\'e-Schlesinger systems}

\paragraph{Integrable Euler-Arnold tops.} For $S\in\Mat$ being a matrix
of dynamical variables define
the Lax pair as
 \beq{q14}
  \begin{array}{c}
  \displaystyle{
  L(S,z)=\frac{1}{N}\,\tr_2\Big(r_{12}(z){S}_2\Big)\,,\qquad M(S,z)=\frac{1}{N}\,\tr_2\Big(m_{12}(z){S}_2\Big)\,,\qquad
  {S}_2=1_N\otimes S\,,
 }
 \end{array}
 \eq
where $\tr_2$ is the trace over the second tensor component, i.e. for
 \beq{q141}
  \begin{array}{c}
  \displaystyle{
  r_{12}(z)=\sum\limits_{ijkl=1}^N r_{ij,kl}(z)E_{ij}\otimes E_{kl}\,,
   }
 \end{array}
 \eq
we have
 \beq{q142}
  \begin{array}{c}
  \displaystyle{
  L(S,z)=\frac{1}{N}\,\tr_2\Big(r_{12}(z){S}_2\Big)=\frac{1}{N}\sum\limits_{ijkl=1}^N r_{ij,kl}(z)E_{ij}S_{lk}\,.
 }
 \end{array}
 \eq
Then the Lax equation (\ref{q05}) holds true identically in spectral parameter $z$ on the
equations of motion (\ref{q06}) with
 \beq{q15}
  \begin{array}{c}
  \displaystyle{
  J(S)=\frac{1}{N}\,\tr_2\Big(m_{12}(0){S}_2\Big)=M(S,0)\,.
 }
 \end{array}
 \eq
The proof of the above statement is based on the usage of (\ref{q57}). One should multiply both
sides of this identity by $S_2S_3$ and compute the trace
 $\tr_{2,3}$ over the second and the third tensor components. This leads to
$[L(S,z),M(S,z)]=L([S,J(S)],z)$, see \cite{LOZ16}.

\paragraph{Heat equation and monodromy preserving equations.} Suppose $R$-matrix $R_{12}^\hbar(z)$, satisfying
the AYBE and the required additional properties, depends also on a parameter $\tau$ in way that
the following heat equation holds true:
 \beq{q60}
  \begin{array}{c}
  \displaystyle{
  2\pi\imath \p_\tau R_{12}^\hbar(z)=\p_\hbar\p_zR_{12}^\hbar(z)\,.
 }
 \end{array}
 \eq
For example, this
 equation holds true for the elliptic Baxter-Belavin $R$-matrix written as given in (\ref{q38}).

For $N=1$ we have
 \beq{q61}
  \begin{array}{c}
  \displaystyle{
  2\pi\imath \p_\tau \phi(\hbar,z)=\p_\hbar\p_z \phi(\hbar,z)\,,
 }
 \end{array}
 \eq
which follows from the heat equation $4\pi\imath\p_\tau\vth(z)=\p_z^2\vth(z)$\,.

Plugging the quasi-classical expansion (\ref{q13}) into (\ref{q60}) we obtain
 \beq{q62}
  \begin{array}{c}
  \displaystyle{
  2\pi\imath \p_\tau r_{12}(z)=\p_z m_{12}(z)\,.
 }
 \end{array}
 \eq
Then the described above construction of integrable Euler-Arnold top can be naturally
extended to the non-autonomous version. Namely, the
non-autonomous\footnote{It is non-autonomous since $J(S)$ explicitly depend on
the time variable $\tau$, see (\ref{q43}).} Euler-Arnold top
 \beq{q63}
  \begin{array}{c}
  \displaystyle{
  2\pi\imath \p_\tau S=[S,J(S)]\,.
 }
 \end{array}
 \eq
with $J(S)$ (\ref{q121}) is equivalent to the
monodromy preserving equation
 \beq{q64}
  \begin{array}{c}
  \displaystyle{
  2\pi\imath \p_\tau L(S,z)-\p_z M(S,z)=[L(S,z),M(S,z)]
 }
 \end{array}
 \eq
with the same matrices $L(S,z)$, $M(S,z)$ (\ref{q14}) as in the Lax equation.

\paragraph{Gaudin models.} The Gaudin model is defined by the Lax matrix
 \beq{q65}
  \begin{array}{c}
  \displaystyle{
  L^{\hbox{\tiny{G}}}(z)=\frac{1}{N}\sum\limits_{a=1}^n\tr_2\Big(r_{12}(z-z_a){S}^a_2\Big)\,.
 }
 \end{array}
 \eq
The classical Gaudin Hamiltonians are the following functions:
 \beq{q66}
  \begin{array}{c}
  \displaystyle{
 H_a^{\hbox{\tiny{G}}} = \frac{1}{N}\sum\limits_{b:b\neq a}^n
 \tr_{12}\left(r_{12}(z_a-z_b)S^a_1S^{b}_2\right),\quad a=1,...,n\,.
 }
 \end{array}
 \eq
and
 \beq{q67}
  \begin{array}{c}
  \displaystyle{
 H_0^{\hbox{\tiny{G}}} = \frac{1}{2N}\sum\limits_{a,b=1}^n
 \tr_{12}\left(m_{12}(z_a-z_b)S^a_1S^{b}_2\right)
 =
 }
 \\ \ \\
  \displaystyle{
 =\frac{1}{2}\sum\limits_{a=1}^n\tr(S^aJ(S^a))+
 \frac{1}{2N}\sum\limits_{a\neq b}^n
 \tr_{12}\left(m_{12}(z_a-z_b)S^a_1S^{b}_2\right),
 }
 \end{array}
 \eq
where $J(S^a)$ was defined in (\ref{q15}).
The Hamiltonian $H_a^{\hbox{\tiny{G}}}$ provides dynamics
(the Hamiltonian flows) with time variable $t_a$, and the equations of motion are of the form:
 \beq{q68}
  \begin{array}{c}
  \displaystyle{
\frac{d}{dt_a}S^b=[S^b,I^{ba}(S^a)]\,,\quad b\neq a
 }
 \\ \ \\
   \displaystyle{
\frac{d}{dt_a}S^a=\sum\limits_{b:b\neq a}^n [I^{ab}(S^b),S^a]\,,
 }
 \end{array}
 \eq
where we introduced notation
 \beq{q69}
  \begin{array}{c}
  \displaystyle{
I^{ab}\left(S^{b}\right) = \frac{1}{N}\,\tr_{2}\left(r_{12}(z_a-z_b)S^{b}_2\right)\,.
 }
 \end{array}
 \eq
The equation of motion generated by the Hamiltonian $H_0^{\hbox{\tiny{G}}}$ (\ref{q67})
is as follows:
 \beq{q70}
  \begin{array}{c}
  \displaystyle{
\frac{d}{dt_0}S^a=[S^a,J(S^a)]+\sum\limits_{b:b\neq a}^n[S^a,J^{ab}(S^b)]\,,
 }
 \end{array}
 \eq
where
 \beq{q71}
  \begin{array}{c}
  \displaystyle{
J^{ab}\left(S^{b}\right) = \frac{1}{N}\,\tr_{2}\left(m_{12}(z_a-z_b)S^{b}_2\right)\,.
 }
 \end{array}
 \eq
Notice that $J^{aa}(S^a)=J(S^a)$.

\begin{predl}\label{prop31}
Equations of motion (\ref{q68}) and (\ref{q70}) are represented in the Lax form
 \beq{q72}
  \begin{array}{c}
  \displaystyle{
 \frac{d}{dt_a}L^{\hbox{\tiny{G}}}(z)=[L^{\hbox{\tiny{G}}}(z),M^a(z)]\,,\quad a=0,1,..,n
 }
 \end{array}
 \eq
with
 \beq{q73}
  \begin{array}{c}
  \displaystyle{
  M^a(z)=-\frac{1}{N}\tr_2\Big(r_{12}(z-z_a){S}^a_2\Big)=-L(S^a,z-z_a)\,.
 }
 \end{array}
 \eq
and
 \beq{q74}
  \begin{array}{c}
  \displaystyle{
  M^0(z)=\frac{1}{N}\sum\limits_{a=1}^n\tr_2\Big(m_{12}(z-z_a){S}^a_2\Big)\,.
 }
 \end{array}
 \eq
\end{predl}
\begin{proof}
Consider the identity (\ref{q58}) in the form
\beq{q741}
\begin{split}
& r_{12}(z-z_a)r_{13}(z-z_b) - r_{23}(z_a-z_b)r_{12}(z-z_a) + r_{13}(z-z_b)r_{23}(z_a-z_b) = \\ & = m_{12}(z-z_a) + m_{23}(z_a-z_b) + m_{13}(z-z_b)\,.
\end{split}
\eq
Here we imply that $z_a\neq z_b$.
Multiply both sides by $S_2^a,S_3^b$ and taking trace $\tr_{23}$ we get
\beq{q742}
\begin{split}
& L(S^{a}, z-z_a)L(S^{b},z-z_b) = L(S^{a}I^{ab}(S^{b}),z-z_a) + L(I^{ba}(S^{a})S^{b},z-z_b) + \\ & + \frac{\tr(S^{b})}{N} M(S^{a},z-z_a) + \frac{\tr(S^{a})}{N} M(S^{b},z-z_b) + \frac{1}{N}\tr_{23}({m_{23}(z_a-z_b)S^{a}_2S^{b}_3})\,.
\end{split}
\eq
Then
\beq{q743}
\begin{split}
& [L(S^{b},z-z_b), L(S^{a},z-z_a)] = L([I^{ab}(S^{b}), S^{a}],z-z_a) + L([S^{b},I^{ba}(S^{a})],z-z_b)\,.
\end{split}
\eq
From this relation the statement of the Proposition (\ref{q72}) for $a=1,...,n$ follows.

Next, consider the identity (\ref{q55}) in the form
\beq{q75}
\begin{split}
    & [m_{13}(z-z_a), r_{12}(z-z_b)] = [r_{12}(z-z_b), m_{23}(z_b-z_a)] +\\ & + [m_{12}(z-z_b), r_{23}(z_b-z_a)] + [m_{13}(z-z_a), r_{23}(z_b-z_a)]\,.
    \end{split}
\eq
Multiply both sides by $[S_2^a,S_3^b]$ and taking trace $\tr_{23}$ we get
\beq{q76}
\begin{split}
& [L(S^{b},z-z_b), M(S^{a},z-z_a)] = L([S^{b},J^{ba}(S^{a})],z-z_b) + M([S^{b},I^{ba}(S^{a})],z-z_b)- \\ & - M([S^{a},I^{ab}(S^{b})],z-z_a)\,,
\end{split}
\eq
where
 \beq{q77}
  \begin{array}{c}
  \displaystyle{
  M\left(S^{a},z\right) = \frac{1}{N} \tr_{2}\left(m_{12}(z)S^{a}_2\right)\,.
 }
 \end{array}
 \eq
From this relation the statement of the Proposition (\ref{q72}) for $a=0$ follows.
For this purpose one should sum up both sides of (\ref{q76}). Then the terms with $M$-matrices
in the r.h.s. are cancelled out.
\end{proof}

\paragraph{Schlesinger system.} The Schlesinger system describes the monodromy preserving equation (see e.g. \cite{CLOZ}). It can viewed as non-autonomous generalization of Gaudin model. Namely, the Schlesinger Hamiltonians
are exactly the same as was given in (\ref{q66}) and (\ref{q67}). But the time variables are now $z_a$
for (\ref{q66}) and $\tau$ entering the heat equations (\ref{q60}) or (\ref{q62}) for the Hamiltonian
(\ref{q67}). That is equations of motion are now of the form:
 \beq{q78}
  \begin{array}{c}
  \displaystyle{
\frac{d}{dz_a}S^b=[S^b,I^{ba}(S^a)]\,,\quad b\neq a
 }
 \\ \ \\
   \displaystyle{
\frac{d}{dz_a}S^a=\sum\limits_{b:b\neq a}^n [I^{ab}(S^b),S^a]\,,
 }
 \end{array}
 \eq
and
 \beq{q79}
  \begin{array}{c}
  \displaystyle{
2\pi\imath\frac{d}{d\tau}S^a=[S^a,J(S^a)]+\sum\limits_{b:b\neq a}^n[S^a,J^{ab}(S^b)]\,.
 }
 \end{array}
 \eq
Similarly to the Gaudin model these equation are represented in the zero-curvature form.
\begin{predl}
Equations of motion (\ref{q78}) and (\ref{q79}) are represented in the form
 \beq{q80}
  \begin{array}{c}
  \displaystyle{
 \p_{z_a}L^{\hbox{\tiny{G}}}(z)-\p_z M^a(z)=[L^{\hbox{\tiny{G}}}(z),M^a(z)]\,,\quad a=1,..,n
 }
 \end{array}
 \eq
with $M^a(z)$ (\ref{q73}) and
 \beq{q81}
  \begin{array}{c}
  \displaystyle{
 2\pi\imath\p_{\tau}L^{\hbox{\tiny{G}}}(z)-\p_z M^0(z)=[L^{\hbox{\tiny{G}}}(z),M^0(z)]
 }
 \end{array}
 \eq
with $M^0(z)$ (\ref{q74}).
\end{predl}
The proof repeats the one for the Gaudin model. One should also use
(\ref{q62}) for (\ref{q81}) and the obvious property $(\p_z+\p_{z_a})r_{12}(z-z_a)=0$ for (\ref{q80}).




\section{2d integrable systems from AYBE}
\setcounter{equation}{0}


\subsection{Higher rank Landau-Lifshitz equation}
Let us briefly recall the construction of higher rank Landau-Lifshitz model
from\footnote{There is a small difference between \cite{AtZ2} and the
below given construction since in this paper we use the Zakharov-Shabat equation in the form
(\ref{q04}), while in \cite{AtZ2} it is $\p_t U(z)-k\p_x V(z)=[U(z),V(z)]$, and $k$ is fixed as $k=1$.} \cite{AtZ2}.
As was explained previously, $U$-matrix has the same form
as in the finite-dimensional mechanics (\ref{q142}):
 \beq{q821}
  \begin{array}{c}
  \displaystyle{
  U(z)=L(S,z)=\frac{1}{N}\,\tr_2\Big(r_{12}(z){S}_2\Big)\,.
 }
 \end{array}
 \eq
The ansatz for $V$-matrix is based on the $R$-matrix identities and its form in the ${\rm sl}_2$
case known from \cite{Skl}:
 \beq{q822}
  \begin{array}{c}
  \displaystyle{
 V(z)=cL(T,z)-c\p_z L(S,z)+L(SE(S),z)+L(E(S)S,z)\,,
 }
 \end{array}
 \eq
where $c\in\mC$ is a constant and $T\in\Mat$ --  some dynamical matrix-valued variable. It will
be defined below. The constant $c$ appears in the additional property, which we require for the matrix $S$:
 \beq{q833}
  \begin{array}{c}
  \displaystyle{
 S^2=cS\,.
 }
 \end{array}
 \eq
This condition means that
the eigenvalues of the matrix $S$ are equal to either $0$ or $c$. We discuss it in the next subsection.

Let us rewrite the matrix $V$ in a different way. For this purpose we use the identity (\ref{q59}).
Multiply both sides by $S_2S_3$ and taking the trace $\tr_{23}$ we obtain
 \beq{q834}
  \begin{array}{c}
  \displaystyle{
 L^2(S,z)=
 }
 \\ \ \\
   \displaystyle{
 =L(SE(S),z)+L(E(S)S,z)-\p_z L(S^2,z)
 +\frac{2\tr(S)}{N}\,M(S,z)
 +\frac{1_N}{N}\,\tr_{23}\Big(m_{23}(0){S}_2{S}_3\Big)\,,
 }
 \end{array}
 \eq
where the following notation was introduced:
 \beq{q835}
  \begin{array}{c}
  \displaystyle{
E(S)=\frac{1}{N}\,\tr_2\Big(r^{(0)}_{12}{S}_2\Big)\,.
 }
 \end{array}
 \eq
Then the matrix $V$ (\ref{q822}) is written as follows:
 \beq{q836}
  \begin{array}{c}
  \displaystyle{
 V(z)=cL(T,z)+L^2(S,z)
 -2s_0\,M(S,z)
 -\frac{1_N}{N}\,\tr_{23}\Big(m_{23}(0){S}_2 S_3\Big)\,,\quad s_0=\frac{\tr(S)}{N}\,.
 }
 \end{array}
 \eq
Similar calculations lead to identity
 \beq{q837}
  \begin{array}{c}
  \displaystyle{
 [L(S,z),L(T,z)]=-\p_zL([S,T],z)+L([S,E(T)],z)+L([E(S),T],z)\,.
 }
 \end{array}
 \eq
Finally, plugging the $U$-$V$ pair into the Zakharov-Shabat equation we come to
the following equations of motion:
 \beq{q838}
  \begin{array}{c}
  \displaystyle{
  \p_t S=ck\p_x T+k\p_x(SE(S)+E(S)S)-2s_0[S,J(S)]+c[S,E(T)]+c[E(S),T]
 }
 \end{array}
 \eq
 and
 \beq{q839}
  \begin{array}{c}
  \displaystyle{
  -k\p_xS=[S,T]\,.
 }
 \end{array}
 \eq
The latter equation can be solved with respect to $T$ in the case (\ref{q833}).
Namely, using $(S-(c/2)1_N)^2=(c/2)^21_N$ one can easily verify that the following
expression solves (\ref{q839}):
 \beq{q840}
  \begin{array}{c}
  \displaystyle{
  T=-\frac{k}{c^2}\,[S,\p_x S]\,.
 }
 \end{array}
 \eq
Plugging it into (\ref{q838}) we get the final equation of motion for $S$:
 \beq{q841}
  \begin{array}{c}
  \displaystyle{
  \p_t S+\frac{k^2}{c}\,[S,\p^2_x S]-k\p_x\Big(SE(S)+E(S)S\Big)=
  }
  \\ \ \\
  \displaystyle{
  =-2s_0[S,J(S)]-\frac{k}{c}\,[S,E([S,\p_x S])]-\frac{k}{c}\,[E(S),[S,\p_x S]]\,,
 }
 \end{array}
 \eq
where $s_0=\tr(S)/N$. Notice that in the ${\rm sl}_2$ case we have $E(S)=0$ due to vanishing of $r_{12}^{(0)}$-coefficient of expansion of $R$-matrices (\ref{q48}), (\ref{q49}), (\ref{q50}) or (\ref{q51}). Therefore, in ${\rm sl}_2$ case
(\ref{q841}) indeed reduces to the original form of the Landau-Lifshitz equation (\ref{q001}).

\subsection{Special coadjoint orbits}
Let us return back to the condition (\ref{q833}). It means that eigenvalues of the matrix $S$ are fixed to be
either equal to $0$ or to $c$.

When $N-1$ eigenvalues are equal to zero and the last one is equal to $c$, the matrix $S$ becomes a rank one
matrix, that is
 \beq{q842}
  \begin{array}{c}
  \displaystyle{
S=\xi\otimes\eta\,,
 }
 \end{array}
 \eq
where $\xi$ is a $N$-dimensional column-vector, and $\eta$ is a $N$-dimensional row-vector. The scalar product
of these vectors equals $c$:
 \beq{q843}
  \begin{array}{c}
  \displaystyle{
(\eta,\xi)=\eta\xi=c=\tr(S)\,.
 }
 \end{array}
 \eq

In the rank one case some additional identities hold true, and this allows to simplify the equation of motion
(\ref{q841}).

\begin{predl}\label{prop41}
Suppose $S\in\Mat$ is of the form (\ref{q841}) and the coefficient $r_{12}^{(0)}$ of the expansion (\ref{q363})
satisfies the property\footnote{The property (\ref{q844}) holds true for the elliptic $R$-matrix and its degenerations. It follows from the Fourier symmetry property $R_{12}^\hbar(z)=R_{12}^z(\hbar)P_{12}$.}
 \beq{q844}
  \begin{array}{c}
  \displaystyle{
r_{12}^{(0)}=r_{12}^{(0)}P_{12}\,.
 }
 \end{array}
 \eq
 Then
 \beq{q845}
  \begin{array}{c}
  \displaystyle{
SE(S)=0\,.
 }
 \end{array}
 \eq
\end{predl}
\begin{proof}
Besides (\ref{q842}) and (\ref{q844}) we are also going to use the skew-symmetry (\ref{q364})
 \beq{q846}
  \begin{array}{c}
  \displaystyle{
r_{12}^{(0)}=-r_{21}^{(0)}=-P_{12}r_{12}^{(0)}P_{12}\,.
 }
 \end{array}
 \eq
First, consider the following expression
 \beq{q847}
  \begin{array}{c}
  \displaystyle{
 S_1r_{12}^{(0)}S_1=(\xi\otimes\eta)_1\sum\limits_{i,j,k,l=1}^N r^{(0)}_{ij,kl} E_{ij}\otimes E_{kl}
 (\xi\otimes\eta)_1=\sum\limits_{i,j,k,l=1}^N r^{(0)}_{ij,kl}(\xi\otimes\eta)_1(\eta E_{ij}\xi)(E_{kl})_2=
 }
 \\
  \displaystyle{
=S_1\sum\limits_{i,j,k,l=1}^N r^{(0)}_{ij,kl}\tr(S E_{ij})(E_{kl})_2=S_1\tr_3(r^{(0)}_{32}S_3)\,.
 }
 \end{array}
 \eq
Then
 \beq{q848}
  \begin{array}{c}
  \displaystyle{
SE(S)=S_1\tr_2(r_{12}^{(0)}S_2)=\tr_2(S_1r_{12}^{(0)}S_2)\stackrel{(\ref{q844})}{=}\tr_2(S_1r_{12}^{(0)}P_{12}S_2)=
 }
 \\ \ \\
  \displaystyle{
=\tr_2(S_1r_{12}^{(0)}S_1P_{12})\stackrel{(\ref{q847})}{=}\tr_{23}(S_1r_{32}^{(0)}S_3P_{12})
=\tr_{23}(S_1P_{12}r_{31}^{(0)}S_3)=\tr_{3}(S_1r_{31}^{(0)}S_3)\stackrel{(\ref{q846})}{=}0\,.
 }
 \end{array}
 \eq
This finishes the proof.
\end{proof}

In a similar way one can derive more complicated relations. In particular,
the following identity holds \cite{AtZ2}:
 \beq{q849}
 \begin{array}{c}
  \displaystyle{
 c\p_x\Big(SE(S)+E(S)S\Big)+
  [E([S,\p_x S]),S]+[[S,\p_x S],E(S)]=2c[E(\p_x S),S]\,.
 }
 \end{array}
 \eq
Then the equation of motion (\ref{q841}) take the form
 \beq{q850}
  \begin{array}{c}
  \displaystyle{
  \p_t S+\frac{k^2}{c}\,[S,\p^2_x S]+2s_0\, [S,J(S)]+2k[S,E(\p_x S)]=0
 }
 \end{array}
 \eq
or (\ref{q851}).
The latter equation has the Hamiltonian formulation, see \cite{AtZ2}.

\subsection{Principle chiral model}
Now we proceed to the multipole case, and first we consider $U$-matrix with simple poles at two
distinct points $z_a$ and $z_b\neq z_a$.
Introduce $U$-matrix
\beq{q862}
\begin{array}{c}
\displaystyle{
U(z) = L(S^{a},z-z_a) +  L(S^{b},z-z_b)\,.
 }
\end{array}
\eq

\subsubsection{Notations and identities}
Let us summarize notations. For any $S\in\Mat$
\beq{q852}
    \begin{split}
        & L\left(S,z\right) = \frac{1}{N}\, \tr_{2}\left(r_{12}(z)S_2\right)\,,\qquad
        E\left(S\right) = \frac{1}{N}\, \tr_{2}\left(r^{(0)}_{12}S_2\right)\,,
         \\
        & I^{ab}\left(S\right) = \frac{1}{N}\, \tr_{2}\left(r_{12}(z_a-z_b)S_2\right)\,,\qquad
        M\left(S,z\right) = \frac{1}{N}\, \tr_{2}\left(m_{12}(z)S_2\right)\,,
         \\
        & J\left(S\right) = \frac{1}{N}\, \tr_{2}\left(m_{12}(0)S_2\right)\,,\qquad
         J^{ab}\left(S\right) = \frac{1}{N}\, \tr_{2}\left(m_{12}(z_a-z_b)S_2\right)\,.
    \end{split}
\eq
Similarly to derivations in the Landau-Lifshitz model we need a set of identities.
Multiplying both sides of (\ref{q55}) written as
\beq{q855}
\begin{split}
    & \left[m_{13}\left(z-z_a\right), r_{12}\left(z-z_b\right)\right] = \left[r_{12}\left(z-z_b\right), m_{23}\left(z_b-z_a\right)\right] +\\ & + \left[m_{12}\left(z-z_b\right), r_{23}\left(z_b-z_a\right)\right] + \left[m_{13}\left(z-z_a\right), r_{23}\left(z_b-z_a\right)\right]
    \end{split}
\eq
by $S_2S_3$ and taking the trace $\tr_{23}$ we obtain
\beq{q854}
\begin{array}{c}
\displaystyle{
[L(S^{b},z-z_b), M(S^{a},z-z_a)] =
}
\\ \ \\
\displaystyle{
 =L([S^{b},J^{ba}(S^{a})],z-z_b) + M([S^{b},I^{ba}(S^{a})],z-z_b)- M([S^{a},I^{ab}(S^{b})],z-z_a).
 }
\end{array}
\eq
In the same way from (\ref{q58}) one gets
\beq{q856}
\begin{split}
& L(T^{a}, z-z_a)L(S^{b},z-z_b) = L(T^{a}I^{ab}(S^{b}),z-z_a) + L(I^{ba}(T^{a})S^{b},z-z_b) + \\ & + \frac{\tr(S^{b})}{N}\, M(T^{a},z-z_a) + \frac{\tr(T^{a})}{N}\, M(S^{b},z-z_b) + \frac{1}{N}\,\tr_{23}({m_{23}(z_a-z_b)T^{a}_2S^{b}_3})\,,
\end{split}
\eq
and, therefore,
\beq{q857}
\begin{split}
& [L(S^{b},z-z_b), L(T^{a},z-z_a)] = L([I^{ab}(S^{b}), T^{a}],z-z_a) + L([S^{b},I^{ba}(T^{a})],z-z_b).
\end{split}
\eq
Also, from (\ref{q56}) one may deduce
\beq{q8571}
    \begin{split}
        &  \tr_{23}\Big(m_{23}(0)S^{a}_2[I^{ab}(S^{b}),S^{a}]_{3}\Big) + \tr_{23}\Big(m_{23}(0)[I^{ab}(S^{b}),S^{a}]_{2}S^{a}_3\Big)  + \\ &  + \tr_{23}\Big(m_{23}(z_a-z_b)S^{a}_2[S^{b},I^{ba}(S^{a})]_{3}\Big) +  \tr_{23}\Big(m_{23}(z_b-z_a)[S^{b},I^{ba}(S^{a})]_{2}S^{a}_3\Big) = 0\,.
    \end{split}
\eq
The following identities were obtained in a similar manner in \cite{AtZ2}.
The first is
\beq{q858}
\begin{split}
    & -c\partial_{z}L\left(S^{a},z-z_a\right) +  L\left(E\left(S^{a}\right)S^{a},z-z_a\right) + L\left(S^{a}E\left(S^{a}\right),z-z_a\right) = \\ & = L^{2}\left(S^{a},z-z_a\right) - \frac{2 \tr\left(S^{a}\right)}{N}M\left(S^{a}, z-z_a\right) - \frac{1}{N}\,1_N\,\tr_{23}\left({m_{23}(0)S^{a}_2S^{a}_3}\right).
    \end{split}
\eq
The second reads
\beq{q859}
\begin{array}{c}
\displaystyle{
  \left[L\left(S^{a},z-z_a\right) , L\left(T^{a},z-z_a\right)\right] =
  }
 \\ \ \\
\displaystyle{
  =-\partial_{z}L\left(\left[S^{a},T^{a}\right],z-z_{a}\right)
  + L\left(\left[S^{a},E\left(T^{a}\right)\right],z-z_{a}\right) + L\left(\left[E\left(S^{a}\right),T^{a}\right],z-z_{a}\right).
  }
\end{array}
\eq
The third one is
\beq{q860}
\begin{split}
  \left[L\left(S^{a},z-z_a\right), M\left(S^{a}, z-z_a\right)\right] =  L\left(\left[S^{a},J\left(S^{a}\right)\right], z-z_a\right)
    \end{split}
\eq
and the fourth is
\beq{q861}
\begin{split}
& L\left(A, z-z_a\right)L\left(B,z-z_a\right) = L\left(A E\left(B\right),z-z_a\right) + L\left(E\left(A\right)B,z-z_a\right) - \\ & - \partial_{z}L\left(AB,z-z_a\right) + \frac{\tr\left(B\right)}{N} M\left(A,z-z_a\right) + \frac{\tr\left(A\right)}{N} M\left(B,z-z_a\right) + \frac{1_N}{N}\,\tr_{23}\left({m_{23}(0)A_2 B_3}\right)
\end{split}
\eq
for any $A,B\in\Mat$.

\subsubsection{Equations of motion for the first flows}
First we describe direct 1+1 analogues of the dynamical flows (\ref{q68}).
Following \cite{Z11} we choose $V$-matrix for this case in the same form as
it is known in the finite-dimensional case (\ref{q73}):
\beq{q863}
\begin{array}{c}
\displaystyle{
V^a(z) =  -L\left(S^{a},z-z_a\right)\,.
 }
\end{array}
\eq
Plugging it into the Zakharov-Shabat equation
  \beq{q864}
  \begin{array}{l}
  \displaystyle{
 \p_{t_a} U(z)-k\p_x V^a(z)=[U(z),V^a(z)]\,,
 }
 \end{array}
 \eq
we obtain the following equations of motion similarly to the Proposition \ref{prop31}:
 \beq{q865}
 \left\{
  \begin{array}{l}
  \displaystyle{
\p_{t_a}S^b=[S^b,I^{ba}(S^a)]\,,
 }
 \\ \ \\
   \displaystyle{
\p_{t_a}S^a+k\p_x S^a= [I^{ab}(S^b),S^a]\,.
 }
 \end{array}
 \right.
 \eq
In the same way for
\beq{q866}
\begin{array}{c}
\displaystyle{
V^b(z) =  -L(S^{b},z-z_b)
 }
\end{array}
\eq
we have
 \beq{q867}
 \left\{
  \begin{array}{l}
  \displaystyle{
\p_{t_b}S^a=[S^a,I^{ab}(S^b)]\,,
 }
 \\ \ \\
   \displaystyle{
\p_{t_b}S^b+k\p_x S^b= [I^{ba}(S^a),S^b]\,.
 }
 \end{array}
 \right.
 \eq
Finally, consider
\beq{q868}
\begin{array}{c}
\displaystyle{
V(z) = L(S^{a},z-z_a) -L(S^{b},z-z_b)\,.
 }
\end{array}
\eq
Then the equations of motion take the form
 \beq{q870}
 \left\{
  \begin{array}{l}
  \displaystyle{
\p_{t}S^a-k\p_xS^a=-2[S^a,I^{ab}(S^b)]\,,
 }
 \\ \ \\
   \displaystyle{
\p_{t}S^b+k\p_x S^b= 2[S^b, I^{ba}(S^a)]\,.
 }
 \end{array}
 \right.
 \eq
This one system is a generalization of the equations (\ref{q03}).
For the rational XXX $r$-matrix $r_{12}(z)=P_{12}/z$ the equations (\ref{q03}) are reproduced from (\ref{q870}).

\subsubsection{Equations of motion for the second flows}
The 1+1 generalization of the Gaudin flow (\ref{q70}), (\ref{q74}) is more complicated.
In fact, for the single marked point case it is the one described as the higher rank Landau-Lifshitz
model. Here we extend it to two marked points.
Based on results of \cite{Z11} and the construction of the
Landau-Lifshitz model (\ref{q822}) consider the following ansatz for $V$-matrix:
\beq{q869}
\begin{array}{c}
\displaystyle{
V^a(z) =
}
\\ \ \\
\displaystyle{
=-c\partial_{z}L\left(S^{a},z-z_a\right) + c L\left(T^{a},z-z_a\right) +  L\left(E\left(S^{a}\right)S^{a},z-z_a\right) + L\left(S^{a}E\left(S^{a}\right),z-z_a\right)\,,
 }
\end{array}
\eq
where $T^a\in\Mat$ is an auxiliary matrix to be defined.

\begin{theor}\label{th1}
Suppose $S^a$ is the matrix described by the special coadjoint orbit
\beq{q875}
\begin{array}{c}
\displaystyle{
(S^a)^2=cS^a\,.
 }
\end{array}
\eq
The Zakharov-Shabat equations (\ref{q864}) with $U$-matrix (\ref{q862}) and $V$-matrix
(\ref{q869}) provide the following equations of motion:
\beq{q871}
    k\partial_{x}S^{a} = \left[S^{a}, I^{ab}(S^{b}) - T^{a}\right]\,,
\eq
\beq{q872}
\begin{array}{c}
     \partial_{t_a}S^{a} =
    \\ \ \\
    =k\partial_{x}(E(S^{a})S^{a} + S^{a}E(S^{a})) + kc\partial_{x}T^{a} + c\left[S^{a},E(T^{a})\right] + c\left[E(S^{a}),T^{a}\right] + c\left[I^{ab}(S^{b}),T^{a}\right] -
    \\ \ \\
    - \frac{2\tr(S^{a})}{N}\left[S^{a},J(S^{a})\right] + S^{a}E(\left[I^{ab}(S^{b}),S^{a}\right]) + E(\left[I^{ab}(S^{b}),S^{a}\right])S^{a} + E(S^{a})\left[I^{ab}(S^{b}),S^{a}\right] +
    \\ \ \\
    + \left[I^{ab}(S^{b}),S^{a}\right]E(S^{a}) + S^{a}I^{ab}(\left[S^{b},I^{ab}(S^{a})\right]) + I^{ab}(\left[S^{b},I^{ab}(S^{a})\right])S^{a}
\end{array}
\eq
and
\beq{q873}
    \begin{split}
        & \partial_{t_a}S^{b} = c \left[S^{b},I^{ba}(T^{a})\right] - \frac{2 \tr(S^{a})}{N}\left[S^{b},J^{ba}(S^{a})\right] + \Big[S^{b}, \Big(I^{ba}(S^{a})\Big)^2\Big]\,.
    \end{split}
\eq
\end{theor}
\begin{proof}
The proof is based on the set of identities presented in the beginning of this subsection.
Main idea is as follows. For computation of the commutator in the r.h.s. of the Zakharov-Shabat equation
it is useful to rewrite $V$-matrix through (\ref{q858}). This gives
\beq{q874}
    \begin{split}
        & \partial_{t_a}L(S^{a},z-z_a)+\partial_{t}L(S^{b},z-z_b) + kc\partial_{z}\partial_{x}L(S^{a},z-z_a) - k\partial_{x}L(E(S^{a})S^{a},z-z_a) - \\ &-k\partial_{x}L(S^{a}E(S^{a}),z-z_a) - kc\partial_{x}L(T^{a},z-z_a) = c[L(S^{a},z-z_a) , L(T^{a},z-z_a)] +
        \\ & + c[L(S^{b},z-z_b) , L(T^{a},z-z_a)] +
       \\ & + [L(S^{a},z-z_a) , L^{2}(S^{a},z-z_a) - \frac{2 \tr(S^{a})}{N}M(S^{a}, z-z_a)] +
        \\ & + [L(S^{b},z-z_b) , L^{2}(S^{a},z-z_a) - \frac{2 \tr(S^{a})}{N}M(S^{a}, z-z_a)]\,.
    \end{split}
\eq
Next one should use the identities (\ref{q854})-(\ref{q857}) and (\ref{q859})-(\ref{q861}).
The terms proportional to $1_N$ are cancelled out via (\ref{q8571}).
After cumbersome calculations one gets the statement.
The equation (\ref{q871}) arises in the second order pole -- as coefficient behind $(z-z_a)^{-2}$.
Equations (\ref{q872}) are (\ref{q873}) come from the poles $(z-z_a)^{-1}$ and $(z-z_b)^{-1}$
respectively.
\end{proof}

Due to the property (\ref{q875}) the equation (\ref{q871}) can be solved similarly to (\ref{q840}):
\beq{q876}
\begin{array}{c}
\displaystyle{
T^{a} = -c^{-2}k\left[S^{a},\partial_{x}S^{a}\right] + I^{ab}(S^{b})\,.
 }
\end{array}
\eq
Consider two examples.
\paragraph{Example: XXX rational $r$-matrix.}
Then
\beq{q877}
\begin{array}{c}
\displaystyle{
V^a(z)= c\frac{S^{a}}{(z-z_{a})^{2}}-\frac{c^{-1}k}{z-z_{a}}\left[S^{a},\partial_{x}S^{a}\right] + \frac{cS^{b}}{(z-z_{a})(z_{a}-z_{b})}\,.
 }
\end{array}
\eq
This case was studied in \cite{Z11}.

\paragraph{Example: ${\rm sl}_2$ elliptic $r$-matrix.}
Then
\beq{q878}
\begin{array}{c}
\displaystyle{
V^a(z)=
}
\\
\displaystyle{
=c \sum_{\alpha = 1}^{3}\sigma_{\alpha}\left(\varphi_{\alpha}(z_{a}-z_{b})\varphi_{\alpha}(z-z_{a})S^{a}_{\alpha} - \partial_{z}\varphi_{\alpha}(z-z_{a})S^{a}_{\alpha} - c^{-2}k \varphi_{\alpha}(z-z_{a})\left[S^{a},\partial_{x}S^{a}\right]_{\alpha} \right)\,.
 }
\end{array}
\eq
This case was studied in \cite{Z11} as well.

\subsubsection{Minimal orbits}
The Landau-Lifshitz equation was simplified in the case of the
''minimal orbit'' (\ref{q842}). Let us study this possibility for the obtained equation
(\ref{q873}).
Suppose
 \beq{q8781}
  \begin{array}{c}
  \displaystyle{
S^a=\xi^a\otimes\eta^a\,,\quad S^b=\xi^b\otimes\eta^b\,,
 }
 \end{array}
 \eq
where $\xi^{a,b}$ are  $N$-dimensional column-vectors, and $\eta^{a,b}$ are $N$-dimensional row-vectors. The scalar products are assumed to be
 \beq{q8782}
  \begin{array}{c}
  \displaystyle{
\eta^a\xi^a=\eta^b\xi^b=c=\tr(S^a)=\tr(S^b)\,.
 }
 \end{array}
 \eq
Together with the property of $r_{12}^{(0)}$ (\ref{q844}) this leads to several identities \cite{AtZ2}.
For $i=a,b$
 \beq{q8783}
    \begin{split}
       & S^{i}E(S^{i}) =0\,,
    \end{split}
\eq
 \beq{q8784}
    \begin{split}
        & S^{i}E(S^{i}_{x}S^{i})= 0\,,
    \end{split}
\eq
 \beq{q8785}
    \begin{split}
        & S^{i}E(S^{i}S^{i}_{x}) = - c S^{i}_{x} E(S^{i})\,,
    \end{split}
\eq
 \beq{q8786}
    \begin{split}
        & c\partial_{x}(S^iE(S^i)+E(S^i)S^i) + \left[E(\left[S^i,S^i_{x}\right]),S^i\right] + \left[\left[S^i,S^i_{x}\right],E(S^i)\right] = 2c\left[E(S^i_{x}),S^i\right]\,.
    \end{split}
\eq
Besides (\ref{q8783})-(\ref{q8786}) we also several additional identities.
\begin{predl}
Suppose $S^a$ and $S^b$ satisfy the properties (\ref{q8781})-(\ref{q8782})
and $r_{12}^{(0)}$ satisfies (\ref{q844}). Then
\beq{q8800}
   \displaystyle{
   S^{a}E(I^{ab}(S^{b})S^{a})=0\,,
   }
\eq
\beq{q8801}
   \displaystyle{
   S^{a}E(S^{a}I^{ab}(S^{b})) =-S^{a}I^{ab}(S^{b})E(S^{a})\,,
   }
\eq
\beq{q8802}
   \displaystyle{
   E(I^{ab}(S^{b})S^{a})S^{a} =E(S^{a})I^{ab}(S^{b})S^{a}\,.
   }
\eq
\end{predl}
\begin{proof}
The proof is similar to the one presented in Proposition \ref{prop41}. Consider (\ref{q8800}).
Notice first that
\beq{q8803}
   \displaystyle{
   S_{1} r_{12}^{(0)}r_{13}(z_a-z_b)S_{1} =S_{1}\tr_{4}(r_{42}^{(0)}r_{43}(z_a-z_b)S_{4})\,.
   }
\eq
Indeed,
\beq{q8804}
    \begin{split}
        & S_{1} r_{12}^{(0)}r_{13}(z_a-z_b)S_{1} = \\
        & \sum\limits_{i,j,k,l,m,n,f,e=1}^N r_{ijkl}^{(0)}r_{mnfe}(z_a-z_b)(\psi\otimes\eta)_{1}(E_{ij}E_{mn})_{1}(\psi\otimes\eta)_{1}(E_{kl})_{2}(E_{mn})_{3} = \\ &
        =\sum\limits_{i,j,k,l,m,n,f,e=1}^N  r_{ijkl}^{(0)}r_{mnfe}(z_a-z_b)(\psi\otimes\eta)_{1}(\eta E_{ij}E_{mn} \psi)_{1}(E_{kl})_{2}(E_{mn})_{3} =
        \\ &
        = \sum\limits_{i,j,k,l,m,n,f,e=1}^N r_{ijkl}^{(0)}r_{mnfe}(z_a-z_b)(\psi\otimes\eta)_{1}\tr( E_{ij}E_{mn}S)_{1}(E_{kl})_{2}(E_{mn})_{3} =\\
       &  = S_{1}\tr_{4}(r_{42}^{(0)}r_{43}(z_a-z_b)S_{4})\,.
    \end{split}
\eq
Then
\beq{q8805}
    \begin{split}
        &  S^{a}E(I^{ab}(S^{b})S^{a}) = \tr_{23}(S^{a}_{1}r_{12}^{(0)}r_{23}(z_a-z_b)S^{b}_{3}S^{a}_{2})=  \tr_{23}(S^{a}_{1}r_{12}^{(0)}r_{13}(z_a-z_b)S^{a}_{1}P_{12}S^{b}_{3}) =
        \\ & = \tr_{234}(S^{a}_{1}r_{42}^{(0)}r_{43}(z_a-z_b)S^{a}_{4}P_{12}S^{b}_{3}) = \tr_{234}(P_{12}S^{a}_{2}r_{42}^{(0)}r_{43}(z_a-z_b)S^{a}_{4}S^{b}_{3}) =
        \\ & =\tr_{34}(S^{a}_{1}r_{41}^{(0)}r_{43}(z_a-z_b)S^{a}_{4}S^{b}_{3}) = -S^{a}E(I^{ab}(S^{b})S^{a}) = 0\,.
    \end{split}
\eq
The identities (\ref{q8801}) and (\ref{q8802}) are proved in the same way.
\end{proof}

The upper identities allows to simplify the equation (\ref{q873}).
Namely, we get (here we have already substituted $T^a$ from (\ref{q876})):
\beq{q8806}
    \begin{split}
        & \partial_{t_a}S^{a} = kc\partial_{x}(I^{ab}(S^{b})) - c^{-1}k^{2} \left[S^{a}, \partial_{x}^{2}S^{a}\right] - \frac{2tr(S^{a})}{N}\left[S^{a}, J(S^{a})\right] + \\ & +
        c\left[S^{a},E(I^{ab}(S^b))\right] + c\left[E(S^{a}),I^{ab}(S^{b})\right] -  c^{-1}k\left[I^{ab}(S^{b}),\left[S^{a},\partial_{x}S^{a}\right]\right] + \\ &  +
        S^{a}I^{ab}\left(\left[S^{b},I^{ba}(S^{a})\right]\right) +I^{ab}\left(\left[S^{b},I^{ba}(S^{a})\right]\right)S^{a}  + 2k\left[E(S^{a}_{x},S^{a})\right]  +\\
        & +2 E(S^{a})I^{ab}(S^{b})S^{a} -  E(S^{a})S^{a}I^{ab}(S^{b})  -  E\left(S^{a}I^{ab}(S^{b})\right)S^{a}\,.
    \end{split}
\eq
Let us mention that the Hamiltonian structure for the derived equations is unknown. It is an interesting
ad important problem. In fact, even for a single marked point case (i.e. for the Landau-Lifshitz equation (\ref{q841})) it is known for
the special coadjoint orbit (\ref{q842}) only. In this case (\ref{q841}) reduces
to (\ref{q850}).

\subsection{1+1 Gaudin model}
Consider the case of arbitrary number of marked points $n$.
Introduce the $U$-$V$ pair similarly to the previous case:
\beq{q879}
\begin{split}
   & U(z) = \sum_{i=1}^{n}L\left(S^{i},z-z_i\right)\\
   & V^a(z) = -c\partial_{z}L\left(S^{a},z-z_a\right) + c L\left(T^{a},z-z_a\right) +  L\left(E\left(S^{a}\right)S^{a},z-z_a\right) + L\left(S^{a}E\left(S^{a}\right),z-z_a\right).
     \end{split}
\eq
\begin{theor}\label{th2}
Suppose $S^a$ is the matrix described by the special coadjoint orbit (\ref{q875}).
The Zakharov-Shabat equations (\ref{q864}) with $U$-matrix (\ref{q862}) and $V$-matrix
(\ref{q869}) provide the following equations of motion:
\beq{q880}
    k\partial_{x}S^{a} = \left[S^{a}, \sum_{i\neq a}^{n}I^{ab}(S^{b}) - T^{a}\right],
\eq
\beq{q881}
\begin{split}
    & \partial_{t_a}S^{a} = k\partial_{x}(E(S^{a})S^{a} + S^{a}E^(S^{a})) + kc\partial_{x}T^{a} + c\left[S^{a},E(T^{a})\right] +
    \\ & +  c\left[E(S^{a}),T^{a}\right] - \frac{2tr(S^{a})}{N}\left[S^{a},J(S^{a})\right]  +
     \\ & +
    \sum_{i\neq a}^{n}\Big(c\left[I^{ai}(S^{i}),T^{a}\right] + S^{a}E(\left[I^{ai}(S^{i}),S^{a}\right]) + E(\left[I^{ai}(S^{i}),S^{a}\right])S^{a} + E(S^{a})\left[I^{ai}(S^{i}),S^{a}\right]\Big) +
     \\ & +
    \sum_{i\neq a}^{n}\Big(\left[I^{ai}(S^{i}),S^{a}\right]E(S^{a}) + S^{a}I^{ai}(\left[S^{i},I^{ia}(S^{a})\right]) + I^{ai}(\left[S^{i},I^{ia}(S^{a})\right])S^{a}\Big)
\end{split}
\eq
and
\beq{q882}
    \begin{split}
        & \partial_{t_a}S^{i} = c \left[S^{i},I^{ia}(T^{a})\right] - \frac{2 tr(S^{a})}{N}\left[S^{i},J^{ia}(S^{a})\right] + \left[S^{i}, I^{ia}(S^{a})\right]I^{ia}(S^{a}) + I^{ia}(S^{a})\left[S^{i}, I^{ia}(S^{a})\right]
    \end{split}
\eq
for $i\neq a$.
\end{theor}
The proof of the above statement is analogues to the Theorem \ref{th1}.

When the condition (\ref{q875}) holds, the equation (\ref{q880}) is solved with respect to $T^a$:
\beq{q8831}
\begin{array}{c}
\displaystyle{
T^{a} = -c^{-2}k\left[S^{a},\partial_{x}S^{a}\right] + \sum_{i\neq a}^{n} I^{ai}(S^{i})\,.
 }
\end{array}
\eq
\paragraph{Minimal orbits.} The obtained equations can be slightly simplified
using additional identities appearing when $r_{12}^{(0)}$ satisfies (\ref{q844}) and
all residues $S^i$ are rank one matrices as in (\ref{q8781})-(\ref{q8782}).
In this case we have (here we have already substituted $T^a$ from (\ref{q8831}))
\beq{q8833}
    \begin{split}
        & \partial_{t_a}S^{a} =  - c^{-1}k^{2} \left[S^{a}, \partial_{x}^{2}S^{a}\right] - \frac{2tr(S^{a})}{N}\left[S^{a}, J(S^{a})\right] + kc\sum_{j\neq a}^{n}\partial_{x}(I^{aj}(S^{j}))+
        \\
         & + 2k\left[E(S^{a}_{x}),S^{a}\right]+
        c\sum_{j\neq a}^{n}\left[S^{a},E(I^{aj}(S^j))\right] + c\sum_{j\neq a}^{n}\left[E(S^{a}),I^{aj}(S^{j})\right] +
         \\ &
         + \sum_{i\neq a}^{n}\Big(-c^{-1}k\left[I^{ai}(S^{i}),\left[S^{a},\partial_{x}S^{a}\right]\right] + S^{a}I^{ai}\left(\left[S^{i},I^{ia}(S^{a})\right]\right)\Big) +
\\
         & +
        \sum_{i\neq a}^{n}\Big(I^{ai}\left(\left[S^{i},I^{ia}(S^{a})\right]\right)S^{a} + 2 E(S^{a})I^{ai}(S^{i})S^{a} -  E(S^{a})S^{a}I^{ai}(S^{i})  -  E\left(S^{a}I^{ai}(S^{i})\right)S^{a}\Big) +\\&+
        c\sum_{i\neq a}^{n}\left[I^{ai}(S^{i}),\sum_{j\neq a}^{n}I^{aj}(S^{j})\right]
    \end{split}
    \eq
and for $i\neq a$:
\beq{q8832}
    \begin{split}
        & \partial_{t_a}S^{i} = c\sum_{j\neq a}^{n}\left[S^{i}, I^{ia}(I^{aj}(S^{j}))\right] - c^{-1}k\left[S^{i},I^{ia}(\left[S^{a},\partial_{x}S^{a}\right])\right] - \frac{2tr(S^{a})}{N}\left[S^{i},M_{ia}(S^{a})\right] + \\ & +
        I^{ia}(S^{a})\left[S^{i}, I^{ia}(S^{a})\right] + \left[S^{i}, I^{ia}(S^{a})\right]I^{ia}(S^{a}).
    \end{split}
\eq

\subsection{Deformation by the twist function}
As was explained in Section 2, $U$-matrix can be multiplied by some rational function.
It comes from solution of (\ref{q310}). The function $k(z)$ may have poles and zeros, which do not
coincide with positions of marked points $z_a$. Then ${\ti U}(z)$ acquires additional poles.
Consider the simplest example.
Let $U(z)$ be the $U$-matrix of the rational Heisenberg model:
\beq{q883}
\begin{array}{c}
\displaystyle{
U(z)=\frac{S}{z-z_1}
 }
\end{array}
\eq
and
\beq{q884}
\begin{array}{c}
\displaystyle{
\frac{k}{k(z)}=\frac{z-w_1}{z-y_1}
 }
\end{array}
\eq
Then the transformed $U$-matrix
\beq{q885}
\begin{array}{c}
\displaystyle{
{\ti U}(z)=\frac{{\ti S}^1}{z-z_1}+\frac{{\ti S}^2}{z-y_1}
 }
\end{array}
\eq
has two poles with linearly dependent residues
\beq{q886}
\begin{array}{c}
\displaystyle{
{\ti S}^1=\frac{z_1-w_1}{z_1-y_1}\,S^1\,,\quad
{\ti S}^1=-\frac{y_1-w_1}{z_1-y_1}\,S^1\,.
 }
\end{array}
\eq
Therefore, the usage of the non-trivial central charge (or the twist function)
adds poles with dependent residues. In the general case we come to 1+1 Gaudin model
with some dependent residues.

\section{Appendix: the Higgs bundles and Hitchin systems}\label{secA}
\def\theequation{A.\arabic{equation}}
\setcounter{equation}{0}

Here we briefly describe main steps of construction of integrable systems through
the Hitchin approach. Main idea is to perform the Hamiltonian reduction starting from
the Higgs field defined on the corresponding Higgs bundle. The moment map
equation is a constrain which projects the Higgs field to the Lax matrix.
Similarly, in the case of the field theory one should change the Higgs bundle to
the affine Higgs bundle. Then we deal with the Lax connection (the $U$-matrix), and it satisfies
a field analogue of the moment map equation.

\subsection{Preliminaries on loop groups and loop algebras}
Let $\gg$ be a simple complex Lie algebra and $L(\gg)=\gg\otimes\mC(y))$, $(y\in \mC^*)$
 is the loop algebra of Laurent polynomials. Notice that in contrast to (\ref{q04}) in this Section we
 use $y\in\mC^*$ for the loop variable, and it is a multiplicative variable. In the rest of the paper we use $x$ -- an additive real valued variable
 as in (\ref{q04}). That is on a unit circle $y=e^{\imath x}$. Let $(\,,\,)$ be
 an invariant form on $\gg$.
 Define the form on $L(\gg)$
  $$
  \lan X\,,\, Y\ran=\oint(X,Y)\frac{dy}y\,.
  $$
Consider its central extension
\beq{la}
\hat{L}(\gg)=\{(X(x),\kappa)\}\,,~~ \kappa\in\mC\,.
\eq
 The commutator of  $\hat{L}(\gg)$ assumes the form
 $$
[(X_1,k_1),(X_2,k_2)]=([(X_1,X_2)]_0,\lan X_1,\p X_2\ran)\,,~~(\p=\imath y\p_y )\,,
 $$
where $[(X_1,X_2)]_0$ is a commutator on $\gg$,
The dual to $\hat{L}(\gg)$ is the Lie coalgebra
\beq{lc}
\hat{L}^*(\gg)=\{\clY=(Y,{k})\sim({k}\p+Y)\}\,.
\eq
 It is defined through the pairing
 \beq{pa}
(\clX,\clY)=\lan X,Y\ran +\kappa{k}\,,
\eq
where ${k}$ is the cocenter.

Let $G$ be a complex Lie group with the Lie algebra $\gg$.
The corresponding loop group $L(G)$ is the map $L(G)\,:\,\mC^*\to G$:
\beq{logr}
 L(G)=G\otimes \mC(y))=\left\{\,\sum_m g_m y^m\,,~g_m\in G\,\right\}\,.
  \eq
It has the central extension
 $\hat L(G)=\{g(x),\zeta\}$  defined by the
  multiplication
\beq{ml}
(g,\zeta)\times(g',\zeta')=\left(gg',\zeta\zeta' {\mathcal C}(g,g')\right),
\eq
where ${\mathcal C}(g,g')$ is a 2-cocycle on $L(G)$
 providing the associativity of the multiplication.

  The adjoint action of $f\in L(G)$ on   $\hat{L}(\gg)$ is defined as
\beq{aa}
\Ad_f\clX=\Ad_f(X,\kappa)=(f Xf^{-1}\,,\,\kappa+\lan f^{-1}\p f,X\ran)\,.
\eq
 Then the coadjoint  action of $L(G)$
 assumes the form
\beq{ca}
\Ad^*_f\clY=\Ad^*_f(Y,{k})=
(f^{-1}Y f+{k} f^{-1}\p f\,,\,{k})\,.
\eq

\subsection{Briefly on Hitchin systems}

Let $\Si$ be a Riemann surface
and $G$ is a simple,
simply-connected complex Lie group with Lie algebra $\gg^\mC$.
Consider a $G$ vector bundle $E$ over $\Si$.

  The operator $D_{\bA}$ defines a holomorphic structure on $E$. Locally we have
  \beq{da}
  D_{\bA}=\p_{\bz}+\bA\,.
  \eq
  The quotient
  \beq{bun}
  Bun(G)=\{ D_{\bA}\}/\clG(G)\,,
  \eq
where
\beq{ggc}
\clG(G)=C^\infty Maps (\Si\to G)\,,
\eq
  is the moduli space of the holomorphic $G$ bundles over $\Si$.

\emph{The Higgs field} $\Phi$ is
a  section of the bundle of endomorphisms $E\to E\otimes\clK_\Si$,
where $\clK_\Si$ is the canonical class.
The $G$\emph{-Higgs bundle} is the pair
  \beq{hb1}
\clH=(D_{\bA}\,,~\Phi)\,.
\eq
It is the cotangent bundle $T^*\clA$ to the affine space of holomorphic structures
$\clA=\{D_{\bA}\}$ on $E$. Therefore, the Higgs bundle can be endowed with the canonical
symplectic form
  \beq{cf}
\Om=  \int_\Si|d^2z|(\de\Phi\stackrel{\wedge}{,}\de \bA)\,,
\eq
where $(~,~)$ is a fixed invariant form on the Lie algebra $\gg^\mC$,
and $\stackrel{\wedge}{,}$ means the inner product in the symplectic form.

The  group $\clG(G)$ (\ref{ggc})
is the group of symplectomorphisms of ({\ref{cf}). The moment map equation
related to  this action in terms of (\ref{hb1}) takes the form
\beq{cm1}
D_{\bA}\Phi=0\,.
\eq
This equation defines the holomorphic structure on the Higgs bundle $\clH$ (\ref{hb1}).
Let $\clC_0$ be the set of solutions of (\ref{cm1}).
The symplectic quotient
\beq{ms}
\clM^H(G)=\clH//\clG(G)\sim\clC_0/\clG(G)
\eq
is the phase space of the Hitchin integrable systems \cite{Hi,Hi2,LOZ} with the Poisson brackets
 coming from ({\ref{cf}).


\subsection{Affine Hitchin systems}

To pass to 2d integrable systems we consider the affine version of the Higgs bundle.
For this purpose we replace the group $G$ on the group  $\hat L(G)$ (\ref{logr}), (\ref{ml}).

Below the set $\clH(G)$ of fields on the affine Higgs bundles is described. As a result
of reduction with respect to gauge symmetries it becomes the phase space of 2d integrable model.

Let us consider an infinite rank $L(G)$-bundle $E^{aff}$  and a line bundle $\clL$ over $\Si$.
Consider also smooth maps $C^\infty(\Om^{(0,1)}(\Si)\to \hat\gg)=\{(\bA(z,\bz,y), \bar \kappa(z,\bz)\}$
(see (\ref{la})).
The components of this map are connection forms on the bundles $E^{aff}$ and $\clL$:
 \beq{ac}
\nabla_{\bA}=\left(
\begin{array}{c}
 \p_{\bz}+\bA(z,\bz,y) \\
 \p_{\bz}+ \bar \kappa(z,\bz)
\end{array}
\right)\,.
 \eq

\paragraph{Affine Higgs connection.}
Consider the map
$C^\infty(\clK_\Si\to \hat L^*(\gg))$ (\ref{lc}) given by
$$
 {k}(z,\bz)\p +\Phi(z,\bz,y)\,,~~\p=\imath y\p_y\,.
$$
It is the affine Higgs connection. The Higgs field $\Phi(z,\bz,y)$ is a  section of the endomorphisms map
$E^{aff}\to E^{aff}\otimes\clK_\Si$
and ${k}(z,\bz)$ is a distribution of $(1,0)$ forms on $\Si$, dual to the space of smooth
 $(0,1)$ forms.  This is the cocentral charge.

Similarly to the previous consideration, define the symplectic form on the fields $(\Phi,{k},\bar A,k)$:
\beq{bom5}
\Om=\frac{1}{\pi}\int_\Si|d^2z|\Bigl(\lan\de\Phi\stackrel{\wedge}{,}\de \bA\ran
+\de{k}\wedge\de\bar \kappa\Bigr)\,.
\eq

\paragraph{Marked points.}
 Introduce two sets of marked points on $\Si$
 \beq{D}
\clD_1= \{z_a\,,~a=1,\ldots,n\}\,,~~\clD_2=\{w_b\,,~b=1,\ldots,m\}\,.
 \eq
 First consider the set $\clD_1$.
Let us assign the coadjoint orbits to the marked points $z_a$ .
  The coadjoint orbits ${\mathcal O}_a={\mathcal O}(p_a^{(0)},c_a^{(0)})$  of the loop group $L(G)$
  are defined as
\beq{c11}
 {\mathcal O}(p_a^{(0)},c_a^{(0)})=\{
 S^a(p_a^{(0)},c_a^{(0)})=g^{-1}p_a^{(0)}g+c_a^{(0)}g^{-1}\p g\,,~g\in L(G)\}\,.
\eq
The orbit (passing through a point $p_a^{(0)}$, and with the central charge $c_a^{(0)}$)
is the coset space ${\mathcal O}(p^{(0)},c^{(0)})\sim L(G)/H$
 for $c^{(0)}\neq 0$,
 and ${\mathcal O}(p^{(0)},0)\sim L(G)/L(H)$, where $H$ is the Cartan subgroup of $G$.
 Thereby  the orbit can be described by the elements of these quotient:
 \beq{g12}
 S^a\to g_a\in L(G)/H\,,~{\rm or}~g_a\in L(G)/L(H)\,.
 \eq
 The invariants defining the orbit ${\mathcal O}(p^{(0)},c^{(0)})$ are the conjugacy classes of the monodromy operator $M$ corresponding to the operator $c^{(0)}\p+S$.
 Therefore, there
is a one-to-one correspondence between the set of $L(G)$-orbits in
the coalgebra $\hat{L}^*(\gg)$ (\ref{lc})
and the set of conjugacy classes in the group $G$.

\paragraph{Symplectic form.}
The symplectic form $\om_a$ on the orbit $\clO_a$ (\ref{c11}) takes the form
\beq{1.6a}
\om_a=\frac{1}{2\pi\imath}
\oint \de\langle S^a(p_a^{(0)},c_a^{(0)}), g^{-1}\de g \rangle\frac{dy}{y}\,.
\eq
The corresponding Poisson brackets are
\beq{pb}
\{S^a_\al(y),S^a_\be(y')\}=\de(y/{y'})\sum\limits_\ga c^\ga_{\al\be}S^a_\ga(x)+c^{(0)}\ka_{\al\be}\p\de(y/{y'})\,,
\eq
where $\ka_{\al\be}$ is the defined above invariant form on $\gg$, $c^\ga_{\al\be}$
are structure constants in $\gg$, and $\al,\be,\ga$ here are indices
in some basis in $\gg$.
The form $\om_a$ (\ref{1.6a}) is invariant under transformations $g\to gf$, $f\in L(G)$. The corresponding moment map
is  $S^a(p_a^{(0)},c_a^{(0)})$.

Define the symplectic form $\Om$ on the space of fields
\beq{cfc1}
\clH(G)=(\bA\,,\,\Phi\,,\,{k}\,,\,\bar \kappa\,,\,\cup_{a=1}^n S^a)
 \,.
\eq
Taking into account (\ref{bom5}) and (\ref{1.6a}) we have
\beq{bomi1}
\Om=\frac{1}{\pi}\int_\Si|d^2z|\Bigl(\lan\de\Phi\stackrel{\wedge}{,}\de \bA\ran
+\de{k}\wedge\de\bar \kappa-\sum_{a=1}^n\om_a\de(z-z_a,\bz-\bz_a)\Bigr)\,,
\eq
$$
\lan\de\Phi\stackrel{\wedge}{,}\de \bA\ran=\f1{2\pi\imath}\oint(\de\Phi\stackrel{\wedge}{,}\de \bA)\frac{dy}y\,.
$$

\paragraph{Gauge symmetries -- symplectomorphisms.}
Define the map
 \beq{stg}
G(\Si)=C^\infty(\Si\to  L(G))=\{f(z,\bz,y)\}\,.
\eq
The gauge group has two components
\beq{gg1}
\hat {\cal G}:=
(G(\Si),\clG_1)\,,~~\clG_1=\exp\Big(\varepsilon_1(z,\bz)\Big)\,,
~\varepsilon_1(z,\bz)\in C^\infty(\Si)\,,
\eq
as well as its Lie algebra
$$
Lie(\hat {\cal G})=L_0\oplus L_1\,,~~L_0=Lie(G(\Si))=\{\ep(z,\bz,y)\}\,,~~
 L_1=C^\infty(\Si)=\{\varepsilon_1(z,\bz)\}\,.
$$
The coalgebras $L^*_0$ and $L^*_1$ are distributions on these spaces.

The actions of  $L_0$ and $L_1$ on the fields take the form
\beq{g0}
\begin{array}{ll}
 1.& \de_\ep \bA(z,\bz,x)=-\p_{\bz}\ep+[\ep,\bA]\,, \\
2. & \de_\ep  \bar \kappa(z,\bz)=\lan \bA\p \ep\ran\,.
\end{array}
\eq
 These formulas are the infinite-dimensional version of (\ref{aa}).
 Also we have
 \beq{g1}
\begin{array}{ll}
1.&  \de_{\varepsilon_{1}}\bA(z,\bz,x) =0\,,\\
 2.& \de_{\varepsilon_{1}}\bar \kappa(z,\bz)= -\p_{\bz}\varepsilon_1\,.
\end{array}
 \eq
 Taking into account the coadjoint action (\ref{ca}) we find
 \beq{g2}
 \de_\ep \Phi=  {k}\p \ep+[\Phi,\ep]\,,~~\de_\ep{k}=0\,,
 \eq
\beq{g5}
 \de_{\varepsilon_1}\Phi=0\,, ~~
 \de_{\varepsilon_1} {k}=0\,.
\eq
Define the action of the gauge group on the orbits variables in terms of the group elements:
\beq{g3}
\de_\ep g_a=g_a\ep(z_a,\bz_a)\,,~~\de_{\varepsilon_1}g_a=0\,.
\eq

\paragraph{Symplectic quotient.}
The moment maps
$\mu_j\,:\,\clH(G)\to Lie^*(\hat {\cal G}^{G})$, $j=1,2$ are defined by the actions
(\ref{g0})-(\ref{g3}):
\beq{mom1}
\mu_0=\p_{\bz}\Phi- {k}\p \bA+[\bA,\Phi]-\sum_{a=1}^nS^a\de(z-z_a,\bz-\bz_a)\in L_0^*\,,
\eq
\beq{mom2}
\mu_1=\p_{\bz}{k}\in L_1^*=(C^\infty)^*(\Si)\,.
\eq
Consider the set of the marked points $\clD_2$ (\ref{D}).
Take
\beq{v2}
\mu_1=\mu_1^0\,,~~\mu_1^0=\sum_{b=1}^m s_b\de(z-w_b,\bar z-\bar w_b)\,,
\eq
where
\beq{v3}
\sum_{b=1}^m s_b=0\,.
\eq
The gauge group $\hat\clG(\Si)$ preserves the values of moments
\beq{me}
1.\ \mu_0=0\,,\qquad 2.\ \mu_1=\mu_1^0\,.
\eq
 In this way we come to the symplectic quotient
$$
\clH(G)//\hat\clG(\Si)=\{(\mu_0=0\,,\,\mu_1=\mu_1^0)/\hat \clG(\Si)\}\,.
$$
This is the set of the gauge equivalent solutions of the moment equations (\ref{me}).


\begin{small}

\end{small}


\begin{thebibliography}{99}
\addcontentsline{toc}{section}{References}




\bibitem{AtZ1} K. Atalikov, A. Zotov,
 J. Geom. Phys., 164 (2021) 104161; 	arXiv:2010.14297 [hep-th].

K. Atalikov, A. Zotov,
JETP Letters, 117:8 (2023), 630--634; 	arXiv:2303.08020 [hep-th].

K. Atalikov, A. Zotov,
 Theoret. and Math. Phys., 219:3 (2024), 1004–1017;	arXiv:2403.00428 [hep-th].

\bibitem{AtZ2} K. Atalikov, A. Zotov,
 JETP Lett. 115, 757-762 (2022); 	arXiv:2204.12576 [math-ph].

\bibitem{AtZ3} K. Atalikov, A. Zotov,
Theoret. and Math. Phys., 216:2 (2023), 1083--1103;
	arXiv:2303.02391 [math-ph].



\bibitem{BB}
R.J. Baxter,
  Ann. Phys. 70 (1972) 193--228.

A.A. Belavin,
Nucl. Phys. B, 180 (1981) 189--200.





%

\bibitem{Chered}
I.V. Cherednik,
Theoret. and Math. Phys. 38 (1979), 120–124.

I.V. Cherednik,
Theoret. and Math. Phys. 47 (1981), 422–425.

\bibitem{CLOZ} Yu. Chernyakov, A.M. Levin, M. Olshanetsky, A. Zotov,
J. Phys. A: Math. Gen. 39 (2006) 12083;
arXiv:nlin/0602043 [nlin.SI].

A. M. Levin, M. A. Olshanetsky, A. V. Zotov,
Russian Math. Surveys, 69:1 (2014), 35--118;
	arXiv:1311.4498 [math-ph].

 \bibitem{CY}
K. Costello, M. Yamazaki,
 	arXiv:1908.02289 [hep-th].

\bibitem{FR}
L.D. Faddeev, N.Y. Reshetikhin,
Annals Phys. 167, 227 (1986).

\bibitem{FTbook}
L.D. Faddeev, L.A. Takhtajan,
{\em Hamiltonian methods in the theory of solitons},
Springer-Verlag, 1987.

\bibitem{FK} S. Fomin, A.N. Kirillov,
Advances in geometry; Prog. in Mathematics book series, 172
 (1999) 147--182.







\bibitem{GSZ}
A. Grekov, I. Sechin, A. Zotov,
J. High Energ. Phys. 2019, 81 (2019); 	arXiv:1905.07820 [math-ph].

A. Grekov, A. Zotov,
J. Phys. A, 51 (2018), 315202;
	arXiv:1801.00245 [math-ph].

E. Trunina, A. Zotov,
J. Phys. A, 55:39 (2022), 395202; 	arXiv:2204.06137 [nlin.SI].




\bibitem{Hi}
N. Hitchin,
 Duke Math. J.,{ 54:1}, (1987) 91--114.

\bibitem{Hi2}
N. Hitchin,
Proc. London Math. Soc., { 3} (1987), 59--126.


\bibitem{Krich2} I. Krichever,
Commun. Math. Phys., 229 (2002) 229–-269; arXiv:hep-th/0108110.

A.A. Akhmetshin, I.M. Krichever, Y.S. Volvovski,
Funct. Anal. Appl., 36:4 (2002) 253–-266;\\ arXiv:hep-th/0203192.



\bibitem{Lacroix2} S. Lacroix, A. Wallberg,
J. High Energ. Phys. 2024, 6 (2024);	arXiv:2311.09301 [hep-th].

\bibitem{LL} L.D. Landau, E.M. Lifshitz,
Phys. Zs. Sowjet., 8 (1935) 153--169.

\bibitem{LOZ} A. Levin, M. Olshanetsky, A. Zotov,
Commun. Math. Phys.  236 (2003) 93--133;     arXiv:nlin/0110045.


\bibitem{LOZ14} A. Levin, M. Olshanetsky, A. Zotov,
JHEP 07 (2014) 012; arXiv:1405.7523 [hep-th].

A. Levin, M. Olshanetsky, A. Zotov,
JHEP 10 (2014) 109; arXiv:1408.6246 [hep-th].

\bibitem{LOZ15}  A.M. Levin, M.A. Olshanetsky, A.V. Zotov,
    Theoret. and Math. Phys., 184:1 (2015), 924--939;\\	arXiv:1501.07351 [math-ph].


\bibitem{LOZ16}
A. Levin, M. Olshanetsky, A. Zotov,
J. Phys. A: Math. Theor. 49:39 (2016) 395202; arXiv:1603.06101.

 T. Krasnov, A. Zotov, Annales Henri Poincare, 20:8 (2019)
2671--2697;
 arXiv:1812.04209 [math-ph].



\bibitem{LOZ22} A. Levin, M. Olshanetsky, A. Zotov,
Eur. Phys. J. C 82, 635 (2022); 	arXiv:2202.10106 [hep-th].






\bibitem{Pol}
 A. Polishchuk, 
Advances in Mathematics 168:1 (2002)   56–-95; arXiv:math/0008156 [math.AG].

\bibitem{Pol2}
A. Polishchuk,
Algebra, Arithmetic, and Geometry, Progress in Mathematics book
series, Volume 270,  (2010) 573--617; arXiv:math/0612761 [math.QA].

T. Schedler,
Mathematical Research Letters, 10:3 (2003) 301--321;
arXiv:math/0212258 [math.QA].

\bibitem{STS} A.G. Reiman, M.A. Semenov-Tian-Shansky,
Zap. Nauchn. Sem. LOMI, 150 (1986) 104–-118; Journal of Soviet Mathematics, 46 (1989) 1631–-1640.



\bibitem{SZ}
I. Sechin, A. Zotov,
Phys. Lett. B, 781 (2018), 1–7; arXiv:1801.08908 [math-ph].

M. Matushko, A. Zotov,
Ann. Henri Poincar\'e, 24 (2023), 3373–3419;
	arXiv:2201.05944 [math.QA].

M. Matushko, A. Zotov,
Nonlinearity, 36:1 (2023), 319;
	arXiv:2202.01177 [math-ph].

	M. Matushko, A. Zotov,
Nuclear Phys. B, 1001 (2024), 116499;
 	arXiv:2312.04525 [math-ph].


\bibitem{SmZ} A.V. Zotov, A.V. Smirnov,
Theoret. and Math. Phys., 177:1 (2013) 1281–-1338.

E.S. Trunina, A.V. Zotov,
Theoret. and Math. Phys., 209:1 (2021), 1331--1356;
arXiv:2104.08982 [math-ph].

\bibitem{Skl} E.K. Sklyanin,
Preprint LOMI, E-3-79. Leningrad (1979).






\bibitem{Vicedo} B. Vicedo,
Lett. Math. Phys. 111, 24 (2021); 	arXiv:1908.07511 [hep-th].

V. Caudrelier, M. Stoppato, B. Vicedo,
Commun. Math. Phys. 405, 12 (2024); 	arXiv:2201.08286 [nlin.SI].

\bibitem{Weil} A. Weil, {\em Elliptic functions according to Eisenstein and
Kronecker}, Springer-Verlag, 
 (1976).

 D. Mumford, {\em Tata Lectures on Theta I, II},
Birkh\"auser, Boston, Mass. (1983, 1984).



\bibitem{ZM}  V.E. Zakharov, A.V. Mikhailov,
Soviet Phys. JETP 74 (1978), 1953--1973.



\bibitem{ZaSh} V.E. Zakharov, A.B. Shabat,
 Soviet physics JETP 34:1 (1972) 62--69.

V.E. Zakharov, A.B. Shabat,
Funct. Anal. Appl., 8:3 (1974), 226–-235.

V.E. Zakharov, A.B. Shabat,
Funct. Anal. Appl., 13:3 (1979), 166–-174.







\bibitem{ZZ} A. Zabrodin, A. Zotov,
JHEP 07 (2022) 023;	arXiv: 2107.01697 [math-ph].


\bibitem{Z16}
A.V. Zotov,
 Theoret. and Math. Phys., 189:2 (2016), 1554--1562; arXiv:1511.02468 [math-ph].

 A.V. Zotov, 
 Theoret. and Math. Phys., 197:3 (2018), 1755--1770;



\bibitem{Z11} A.V. Zotov,
SIGMA 7 (2011), 067; 	arXiv:1012.1072 [math-ph].

A. Levin, M. Olshanetsky, A. Zotov,
	Nuclear Physics B 887 (2014) 400--422; 	arXiv:1406.2995 [math-ph].

\bibitem{Z24} A. Zotov,
	J. Phys. A: Math. Theor. 57 (2024) 315201; 	arXiv:2404.01898 [hep-th].

A. Zotov,
{\em On the field analogue of elliptic spin Calogero-Moser model: Lax pair and equations of motion},
Func. Anal. Appl. (2025) to appear;	arXiv:2407.13854 [nlin.SI].

\end{thebibliography}
\end{document}